\documentclass[12pt,a4paper,oneside, reqno]{article}

\usepackage{a4wide,xspace}

\usepackage{authblk}

\usepackage{url}
\usepackage{amssymb,amsmath,amsthm}

\usepackage{color}

\usepackage{tikz}
\usetikzlibrary{decorations.markings}

\usepackage{psfrag}

\definecolor{firebrick}{rgb}{0.7,0.13,0.13}

\newcommand{\hosszu}[1]{\textcolor{black}{#1}}

\newcommand{\new}{\text{new}}

\newcommand{\FTL}{\mathsf{FTL}}
\newcommand{\fvp}{\mathsf{P}}  

\newcommand{\inc}{\mathsf{in}}
\newcommand{\out}{\mathsf{out}}

\newcommand{\taxis}{\mathsf{t\text{-}axis}}
\newcommand{\ray}{\mathsf{ray}}

\newcommand{\fm}{\mathsf{fm}}
\newcommand{\fmom}{\mathsf{fm}}
\newcommand{\timed}{\mathsf{time}}

\newcommand{\sqspace}{\mathsf{space}}

\newcommand{\sqspeed}{\mathit{v}}
\newcommand{\vel}{\mathbf{v}}

\newcommand{\vo}{{\bar o}}
\newcommand{\vx}{{\bar x}}
\newcommand{\vy}{{\bar y}}
\newcommand{\vz}{\bar z}
\newcommand{\vu}{\bar u}
\newcommand{\vvvv}{\bar v}
\newcommand{\leteq}{\mbox{$:=$}}
\newcommand{\gyok}{\sqrt{\phantom{n}}}
\newcommand{\de}{:=}

\newcommand{\coll}{\mathsf{coll}}
\newcommand{\poscoll}{\mathsf{coll}}
\newcommand{\inecoll}{\mathsf{inecoll}}

\newcommand{\defiff}{\ \stackrel{\text{def}}{\Longleftrightarrow}\ }
\newcommand{\lland}{\;\land\;}
\newcommand{\llor}{\;\lor\;}

\newcommand{\vv}{\mathit{v}}

\newcommand{\m}{\ensuremath{\mathsf m}}
\newcommand{\Ip}{\ensuremath{\mathsf{Ip}}}
\newcommand{\M}{\ensuremath{\mathsf M}}
\newcommand{\IOb}{\ensuremath{\mathsf{IOb}}} 
\newcommand{\IB}{\ensuremath{\mathsf{Ip}}}
\newcommand{\B}{\ensuremath{\mathit{B}}} 
\newcommand{\Ph}{\ensuremath{\mathsf{Ph}}} 
\newcommand{\Q}{\ensuremath{\mathit{Q}}} 
\newcommand{\W}{\ensuremath{\mathsf{W}}} 
\newcommand{\ev}{\ensuremath{\mathsf{ev}}} 
\newcommand{\Id}{\ensuremath{\mathsf{Id}}} 
\newcommand{\ax}[1]{\textcolor{axcolor}{\ensuremath{\mathsf{#1}}}} 
\newcommand{\wl}{\ensuremath{\mathsf{wl}}}

\newcommand{\w}{\ensuremath{\mathsf{w}}}

\definecolor{thmcolor}{rgb}{0,0,.4} 
\definecolor{remarkcolor}{rgb}{0,.2,0} 
\definecolor{proofcolor}{rgb}{.4,0,0} 
\definecolor{quecolor}{rgb}{.2,.2,0} 
\definecolor{axcolor}{rgb}{.3,0,.3}
\definecolor{thmbgcolor}{rgb}{0.9,0.9,1} 
\definecolor{rmbgcolor}{rgb}{0.9,1,0.9} 
\definecolor{proofbgcolor}{rgb}{1,0.9,0.9}

\definecolor{qcolor}{rgb}{0,0.4,0}
\definecolor{lqcolor}{rgb}{0,0.6,0}
\definecolor{phcolor}{rgb}{.6,0,0}
\definecolor{lphcolor}{rgb}{.7,0,0}
\definecolor{evcolor}{rgb}{.5,.3,.2}
\definecolor{obcolor}{rgb}{0,0.4,.5}
\definecolor{lobcolor}{rgb}{0,0.6,.75}
\definecolor{iobcolor}{rgb}{0,0,.6}
\definecolor{liobcolor}{rgb}{0,0,.7}
\definecolor{axbgcolor}{rgb}{1,.7,1}

\theoremstyle{definition} 
\newtheorem{thm}{\colorbox{thmbgcolor}{\textcolor{thmcolor}{Theorem}}}[section] 

\newtheorem{cor}[thm]{\colorbox{thmbgcolor}{\textcolor{thmcolor}{Corollary}}} 

\begin{document}

\title{The Existence of Superluminal Particles is Consistent with
Relativistic Dynamics} 

\author{Judit X. Madar\'asz and Gergely Sz\'ekely}
\affil{\small Alfr{\'e}d R{\'e}nyi Institute of Mathematics, Hungarian Academy of 
Sciences, Budapest P.O.Box 127, H-1364, Hungary. 
E-mails: madarasz.judit@renyi.mta.hu, szekely.gergely@renyi.mta.hu}

\date{\small\today}

\maketitle

\begin{abstract} Within an axiomatic framework, we prove that the existence
of faster than light (FTL) particles is consistent with (does not
contradict) the dynamics of Einstein's special relativity.  The proof goes
by constructing a model of relativistic dynamics where FTL particles can
move with arbitrary speeds.  To have a complete picture, we not only
construct an appropriate model but explicitly list all the basic assumptions
(axioms) we use. \end{abstract}

\noindent \emph{Keywords:} 
{special relativity, 
dynamics,
superluminal motion, tachyons, axiomatic method, first-order logic}

\section{Introduction}

From time to time certain experiments (such as OPERA 2011, MINOS 2007, etc.)
appear suggesting that there may be faster than light (FTL) particles. 
Almost all of these experiments turned out to be erroneous so far.  However,
  the tendency that these experiments usually turn out to be erroneous gives
  us no guarantee that there will be no experiment in
the future justifying the existence of FTL particles.  Also Recami's
  recent overview \cite{Recami08} contains some experimental sectors of
  physics still suggesting the existence of FTL objects.

Anyway,  if we have a reliable experiment showing the existence of FTL
particles, we have to rebuild or modify all the theories inconsistent with
(contradicting) FTL motion.  Weinberg--Salam theory is a good
  example of such a theory because it implies the impossibility of
  FTL motion \cite{Meszaros2}.

In this paper, we show that the particle dynamics of Einstein's special
  relativity would survive any experiment showing the existence of FTL
  objects because it is logically consistent with their existence.  The only
framework for investigating the consistency of a statement with a theory is
the axiomatic framework of mathematical logic. Therefore, we investigate the
consistency of FTL particles in the framework of mathematical logic. 

The investigation of FTL motion goes back to pre-relativistic times, see,
e.g., \cite{Fro94}, \cite[\S 3]{Rec86}.  Since 1905 it has generally been
believed that the nonexistence of FTL particles is a direct consequence of
special theory of relativity.  Since Tolman's antitelephone argument
\cite{tolman}, several paradoxes concerning causality violations and FTL
particles have appeared, and since the 1950s great many papers have been
published on theories for FTL particles as well as on possible resolutions
of the paradoxes, see, e.g., \cite{Arntzenius, BiDeSu62, ChaSil, Geroch,
GMMR, HC12, JW12, Je12, Nikolic, Pe13, Rec86,recami-ftl,Rec09, RZD04,
selleri-ftl,  Su70, SS86,  ZRB10,ZRB12},
 and references therein.

Since causality paradoxes are based on changing the past some way,
they are usually resolved by making restrictions on the things that can be
changed in the corresponding situations, see, e.g., Novikov's
self-consistency principle \cite{Nov90,Nov92}.  The possible resolution of
causal paradoxes has an extensive literature.  Moreover, our research
group showed in \cite{CQG}, that FTL motion does not
imply that information can be sent to the past even if we assume
that there are FTL particles moving with arbitrary speeds.
Therefore, instead of investigating the FTL motion based causal
paradoxes,  here we concentrate  only to the more basic question
whether relativistic dynamics allows the existence of massive FTL particles
or not.

To show that relativistic dynamics allows the existence of massive FTL
particles, we have to construct a model of relativistic dynamics where there
are such particles.  However, to have a complete picture, not only the
model construction is important but the basic assumptions (axioms) we take. 
Therefore, we introduce an axiomatic theory of relativistic dynamics
(\ax{SRDyn}) and show that this axiom system has an appropriate model.


As far as we know, apart from ours, none of the theories for FTL particles
in the literature is truly axiomatic in the sense of mathematical logic. 
A key feature of working within a truly axiomatic theory lies in the
fact that within such a theory no tacit assumptions are allowed, all the
assumptions have to be revealed as formal axioms.  This feature is crucial
in investigating consistency questions as well as any other foundational
questions because in these investigations we have to see clearly what is
being assumed and what is not.

In an axiomatic framework similar to the one used here, \cite{FTLconsSR}
shows that the existence of FTL inertial particles does not contradict
(i.e., it is consistent with) special relativistic kinematics.  In
other words, there is a model of relativistic kinematics containing FTL
particles.  This means exactly that the existence of FTL particles is
logically independent of relativistic kinematics because, of course,
there is also a model of relativistic kinematics in which there are no FTL
particles.

In this paper, we show that the existence of massive FTL inertial particles
is logically independent of special relativistic dynamics, too.  This means
that relativistic dynamics implies neither the nonexistence nor the
existence of massive FTL particles; or equivalently both the existence and
the nonexistence of massive FTL particles are consistent with relativistic
dynamics.

This situation is completely analogous to the fact that Euclid's
postulate of parallels is logically independent of the rest of its axioms
(in this case two different consistent theories extending the theory of
absolute geometry are Euclidean geometry and hyperbolic geometry).



Based on Einstein's original 1905 postulates, we formalize the dynamics of
special relativity within an axiomatic framework.  We chose first-order
logic to formulate axioms of special relativity because experience (e.g., in
geometry and set theory) shows that this logic is an adequate logic for
providing axiomatic foundations for a theory.
 
To create any theory of FTL particles, we have to deal with the following
  phenomenon implied already by the kinematics of special relativity.  If an
  observer sees a fusion of two particles in which an FTL particle
  participates, then a fast enough (but slower than light) observer sees
  this fusion as a decay, see Fig.\ref{decay}.
The same example also appears, e.g., in 
\cite{BiDeSu62,Rec86,Su70} and in connection with the phenomenon
\cite{BiDeSu62} says: {\it ``... according to the original criteria, various 
observers must agree on the
identity of physical laws, and not the description of any given phenomenon
...''\/}.
So the existence of FTL particles adds new concepts to the already long list
of observer dependent concepts of relativity theory, namely it is also
observer dependent whether a particle participates in a decay or a fusion.

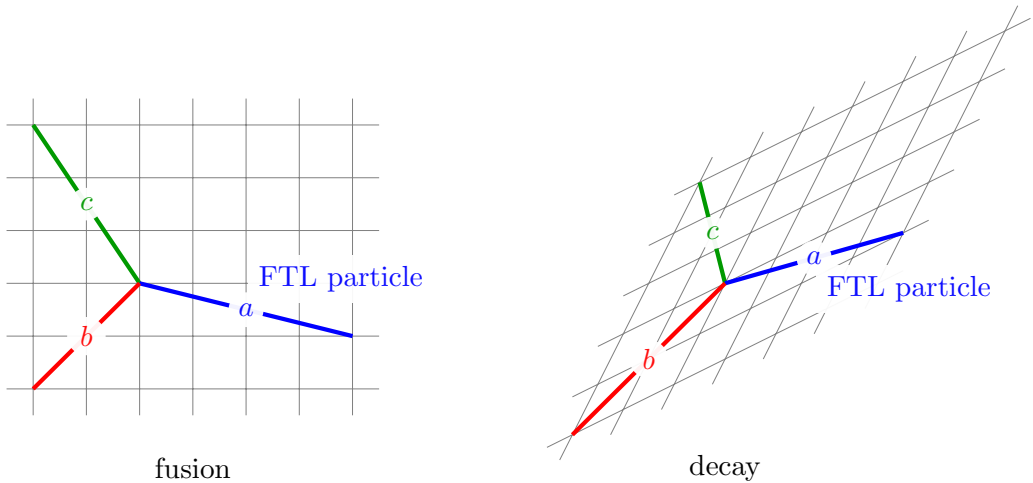
\begin{figure} \centering \small
\begin{tikzpicture}[scale=0.7]
\tikzstyle{nyil}=[ultra thick]
\tikzstyle{cimke}=[fill=white, opacity=0.95, circle, inner sep=2,]
\tikzstyle{boxcimke}=[fill=white, opacity=0.95, inner sep=4,]

\pgfmathsetmacro\ma{1}
\pgfmathsetmacro\pa{-4}
\pgfmathsetmacro\mb{2}
\pgfmathsetmacro\pb{2}
\pgfmathsetmacro\v{-0.5}
\pgfmathsetmacro\gamma{1+\v*\v/2}

\pgfmathsetmacro\pc{\pa+\pb}
\pgfmathsetmacro\mc{\ma+\mb}
\pgfmathsetmacro\MaxAbove{abs(\mc)+0.5}
\pgfmathsetmacro\MaxBelow{-max(abs(\ma),abs(\mb))-0.5}
\pgfmathsetmacro\MaxLeft{-max(abs(\pa),abs(\pb),abs(\pc))-0.5}
\pgfmathsetmacro\MaxRight{max(abs(\pa),abs(\pb),abs(\pc))+0.5}

\begin{scope}[shift={(-\MaxRight -1,0)}]
\coordinate (o) at (0,0);
\coordinate (a) at (\pa,\ma);
\coordinate (b) at (\pb,\mb);
\coordinate (c) at (\pc,\mc);

\draw[help lines] (\MaxLeft+2,\MaxBelow) grid (\MaxRight,\MaxAbove);
	\draw[nyil,green!60!black] (o) -- node[cimke]{$c$} (c);
	\draw[nyil,blue] ([rotate=180]a)-- node[cimke]{$a$} node [boxcimke,above right, yshift=2] {FTL particle}(o);
	\draw[nyil,red] ([rotate=180]b)--node[cimke]{$b$}(o);
\node (D) at (1,\MaxBelow-1) {fusion};
\end{scope}

\begin{scope}[shift={(\MaxRight+1,0)},cm={\gamma,-\v*\gamma,-\v*\gamma,\gamma,(0,0)},scale=0.85]
\coordinate (o) at (0,0);
\coordinate (a) at (\pa,\ma);
\coordinate (b) at (\pb,\mb);
\coordinate (c) at (\pc,\mc);

\draw[help lines] (\MaxLeft+2,\MaxBelow) grid (\MaxRight,\MaxAbove);
	\draw[nyil,green!60!black] (o) -- node[cimke]{$c$} (c);
	\draw[nyil,blue] ([rotate=180]a)-- node[cimke]{$a$} node [boxcimke,below right, yshift=-2] {FTL particle}(o);
	\draw[nyil,red] ([rotate=180]b)--node[cimke]{$b$}(o);
\end{scope}
\node (D) at (\MaxRight+1,\MaxBelow-1) {decay};

\end{tikzpicture}
\caption{A fusion according to one observer can be a decay according to 
another one if there are FTL particles}
\label{decay}
\end{figure}

It is important that Einstein's theory is consistent with the FTL motion of
  particles, but not with the FTL motion of inertial observers (reference
  frames) unless the space is also one dimensional, see
  Corollary~\ref{noftlobsthm}.  If the space is one dimensional, then FTL
  inertial observers can be introduced, see, e.g., \cite[\S 5]{Rec86} or
  \cite[\S 2.7 \& \S3.4]{pezsgo}.

In \cite{SS86} Sutherland and Shepanski introduce FTL reference frames and
  transformations between them.  In the recent paper \cite{HC12}, Hill and
  Cox also introduce transformations between FTL reference frames.  However,
  as Corollary~\ref{noftlobsthm} indicates it, these transformations
  contradict the principle of relativity unless spacetime is two
  dimensional, see \cite{SS86}, \cite{HCnote}.

The structure of this paper is as follows.  In Section~\ref{sec-inf}, we
explain our result and axiomatic framework without going into the details of
formalization.  In Section~\ref{sec-proofinf}, we give the intuitive idea of
the proof of our result.  In Section~\ref{sec-lang}, we recall an axiomatic
framework for dynamics from \cite{dyn-studia}.  In Section~\ref{sec-ax}, we
recall an axiom system and some theorems for kinematics of special
relativity relevant to our present investigation.  In
Section~\ref{sec-dinax}, we present an axiom system for dynamics of special
relativity theory.  The axioms for dynamics are some natural assumptions on
collisions of inertial particles, e.g., conservation of relativistic mass
and linear momentum.  In Section~\ref{sec-main}, within this axiomatic
framework, we formulate and prove our main result, namely that the existence
of FTL inertial particles is independent of dynamics of special relativity,
i.e., we prove that neither the existence nor the nonexistence of FTL
inertial particles follows from the theory, see Theorem~\ref{fotetel}. 
Consequently, it is consistent with dynamics of special relativity that
there are FTL particles.  In Section~\ref{sec-con}, we show an experimental
prediction of Einstein's special relativity on FTL particles, namely that
the relativistic mass and momentum of an FTL particle decrease with the
speed, see also~\cite{BiDeSu62,HC12,Rec86}.

\section{Informal statement of the main result}
\label{sec-inf}

To prove our statement on the existence of massive FTL inertial particles, 
we present an axiom system \ax{SRDyn} which is a formalized version of
Einstein's special relativistic dynamics, see p.\pageref{dinamika}.
Informally, \ax{SRDyn} contains the following axioms for kinematics
(see Fig.\ref{d13} on p.\pageref{d13}):
\begin{enumerate}
\item Principle of relativity (Einstein's first postulate): The
  same laws of nature holds for all inertial observers (reference
  frames) (see
  \ax{SPR^+} on p.\pageref{SPR}).
\end{enumerate}

The second postulate of Einstein literally states that ``Any ray
  of light moves in the stationary system of co-ordinates with the
  determined velocity $c$, whether the ray be emitted by a stationary
  or by a moving body,'' see \cite{einstein}. So our second axiom will
  state the existence of an inertial reference frame according to
  which the speed of light is the same in every direction
  everywhere. However, it is important to note that by the principle
  of relativity, all the reference frames have to have this property
  since there is no distinguished inertial frame of reference.

\begin{enumerate}
\setcounter{enumi}{1}
\item The light axiom (Einstein's second postulate): There is \emph{at least
    one} inertial observer, according to whom all light signals move with the
  same speed (see \ax{AxLight} on p.\pageref{axlight}).
\end{enumerate}

See \cite[Prop.1]{FTLconsSR} for a precise formulation and proof
  of the above intuitively clear argument on that \ax{AxLight} and
  \ax{SPR^+} imply that the speed of light is the same for \emph{all}
  inertial observers.

A benefit of working within a formal axiomatic framework is that
  we have to state explicitly even the most trivial assumptions. This
  is a great help in revealing the tacit assumptions of the
  investigated theory. So now we list some trivial assumptions which
were implicitly assumed by Einstein, as well as by all approaches to
special relativity theory. However, in an axiomatic framework,
  these (or some other) auxiliary axioms are needed to be stated
  explicitly to get back the intended meanings of Einstein's two
  postulates.

\begin{enumerate}
\setcounter{enumi}{2}
\item Physical quantities satisfy some algebraic
  properties of real numbers (see \ax{AxEField} on
  p.\pageref{ax-fd}).
\item Inertial observers coordinatize the same events (see \ax{AxEv}
  on p.\pageref{AxEv}).
\item Inertial observers are stationary according to their own
  coordinate systems (see \ax{AxSelf} on p.\pageref{axself}).

\item Inertial observers (can) use the same units of measurements (see
  \ax{AxSymD} on p.\pageref{AxSymD}).
\end{enumerate}

Theorem~\ref{thm-poi} (see p.\pageref{thm-poi}) justifies that
  the axioms corresponding to the statements above really captures the
  kinematics of special relativity because they imply that the
  transformations between inertial observers (reference frames) are
  Poincar\'e transformations.

In the axioms of \ax{SRDyn} concerning dynamics, we use the notion of
collision of particles. Intuitively, by a possible collision according
to an inertial observer at a coordinate point we mean a set of
incoming and outgoing inertial particles such that the relativistic
mass and linear momentum are conserved, i.e., the sum of the
relativistic masses of the incoming particles coincides with that of
the outgoing ones and the same holds for the linear momenta of the
particles, see Fig.\ref{ppcoll1}. So the conservations of relativistic
mass and linear momentum are built into the definition of the possible
collisions. Inelastic collisions are defined as collisions in which
there is only one outgoing particle.

\begin{figure}
\begin{center}
\small
\psfrag*{t}[r][r]{time}
\psfrag*{inc}[ct][ct]{incoming particles}
\psfrag*{out}[cb][cb]{outgoing particles}
\psfrag*{cp}[lb][lb]{coordinate point}
\psfrag*{space}[lt][lt]{space}
\psfrag*{time}[r][r]{time}

\includegraphics[keepaspectratio,width=0.5\textwidth]{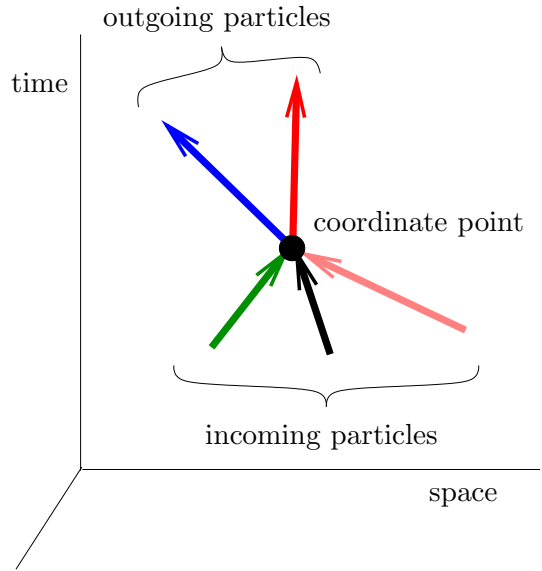}
\caption{Illustration for possible collision;
relativistic mass and linear momentum are conserved}
\label{ppcoll1}
\end{center}
\end{figure}

Now we list the axioms of \ax{SRDyn} concerning dynamics (see
Fig.\ref{d12} on p.\pageref{d12}):
\begin{enumerate}
\setcounter{enumi}{6}
\item The notion of possible collision does not depend on the
inertial observer (see \ax{AxColl_n} on p.\pageref{coll}). 
By the definition of possible collisions, this assumption basically states 
the conservation of relativistic mass and linear momentum.
\item Particles (with given velocities and relativistic
  masses) can be collided inelastically at any coordinate point (see
  \ax{Ax\forall \inecoll} on p.\pageref{axinecoll}).

\item Relativistic masses of slower than light inertial particles depend
only on their speeds (see \ax{AxSpd} on p.\pageref{axspeed}).
\item If the velocities and relativistic masses of two particles coincide for
  one inertial observer then they coincide for all the other inertial
observers,   too (see \ax{AxMass} on p.\pageref{axmass}).

\item Inertial observers can move with any slower than light velocity
 and there are inertial
particles of arbitrary positive relativistic masses and arbitrary
non-FTL velocities (see \ax{AxThEx^+} on
p.\pageref{axthexp}).

\item Every potential collision can be realized (see \ax{Ax\forall Coll} 
on p.\pageref{axminden}).
\end{enumerate}

The main result of this paper is the 
following, see Thm.\ref{fotetel} (p.\pageref{fotetel}):

\begin{quote}
\it The existence of massive FTL inertial particles is consistent with special 
relativistic dynamics \ax{SRDyn}. The nonexistence of FTL inertial
particles
is also consistent with \ax{SRDyn}. Therefore, the existence of massive FTL inertial
particles
is logically independent of \ax{SRDyn}.
\end{quote}

\section{The idea of 
constructing a model for FTL particles}
\label{sec-proofinf}

The main result says that the existence of massive FTL inertial particles is
independent of special relativistic dynamics \ax{SRDyn}.  To prove this
statement, we construct two models (solutions of the axioms) of \ax{SRDyn}
such that there are massive FTL inertial particles in one model and there are
no FTL inertial particles in the other one.

The interesting case is the
construction of the model in which there are massive FTL inertial particles.
Now we turn to explaining the intuitive idea of the construction of
this model.

The key idea is similar to the ideas of
Sudarshan~\cite{Su70}, Recami~\cite{Rec86,recami-ftl},
Bilaniuk et al.~\cite{BiDeSu62} and Arntzenius~\cite{Arntzenius}
using the ``switching-reinterpretation''
principle. The main advantage of our approach over the former ones is that
we formulate an explicit axiom system \ax{SRDyn}
 of relativistic dynamics and we
show by a concrete model construction that this axiom system is consistent
with FTL particles.

To simplify the proof, we use the notion of four-momentum, which is a
defined concept in our framework.  Since we assume that the speed of
light is $1$, the four-momentum of an inertial particle 
is a spacetime vector whose time component is the
relativistic mass and space component is the linear momentum of the
particle, see Fig.\ref{fmom} on p.\pageref{fmom}. Thus in possible
collisions four-momentum is conserved, see Fig.\ref{ppcoll1}.

First we construct the worldview of a distinguished observer having
massive FTL inertial particles. 
For every coordinate point and every nonzero spacetime vector
with nonnegative time component,
we include an incoming and an outgoing inertial particle.
The vectors
will correspond to the four-momenta of the corresponding particles.  
Clearly, there are inertial particles with arbitrary speeds in the worldview of the
distinguished observer, thus there are FTL ones.    

Constructing the worldview of one observer having FTL 
particles is easy. 
The nontrivial part of 
our construction is to construct a worldview of
  observers moving with respect to this observer and associating
  relativistic masses to all the possible particles in the moving
  frame such that all the axioms of \ax{SRDyn} are satisfied. By
  Theorem~\ref{thm-poi}, the worldviews are transformed by Poincar\'e
  transformations. The relativistic masses of slower than light
   particles also have to be transformed in accordance with the
  corresponding Poincar\'e transformation. So the question whether our
  construction can or cannot be finished depends on whether we can
  associate appropriate relativistic masses to FTL  particles.

\begin{figure}
\psfrag*{P}[lt][lt]{\shortstack[l]{distinguished observer\\ $A,B,C$ 4-momenta} }
\psfrag*{n}[lt][lt]{\shortstack[l]{new observer\\ $A^\new,B^\new,C^\new$
4-momenta} }
\psfrag*{f}[ct][ct]{fusion}
\psfrag*{d}[ct][ct]{decay}
\psfrag*{A}[lb][lb]{$A$}
\psfrag*{B}[rb][rb]{$B$}
\psfrag*{C}[rb][rb]{$C$}
\psfrag*{A'}[lt][lt]{$A'$}
\psfrag*{B'}[rb][rb]{$B'$}
\psfrag*{C'}[lb][lb]{$C'$}
\psfrag*{AA}[rt][rt]{$A$}
\psfrag*{t1}[c][c]{$A+B=C$}
\psfrag*{t2}[c][c]{$A'+B'=C'$}
\psfrag*{t3}[c][c]{$B^\new=A^\new+C^\new$}
\psfrag*{A''}[lt][lt]{$A^\new$}
\psfrag*{B''}[rb][rb]{$B^\new$}
\psfrag*{C''}[lb][lb]{$C^\new$}
\psfrag*{Aa}[lt][lt]{$A^\new$}
\psfrag*{Bb}[rb][rb]{$B^\new$}
\psfrag*{Cc}[lb][lb]{$C^\new$}
\psfrag*{c}[lb][lb]{$c$}
\psfrag*{a}[lt][lt]{$a$}
\psfrag*{b}[lt][lt]{$b$}

\psfrag*{p}[cb][cb]{Poincar\'e tr.}
\includegraphics[keepaspectratio,width=\textwidth]{fig3}
\caption{Illustration for model construction}
\label{bev}
\end{figure}


Since mass and four-momentum determine each other it is enough 
to concentrate to the transformation of four-momenta.
To understand why and how four-momenta have to transform differently for 
FTL particles let us consider the following situation.

Let $a$, $b$ and $c$ be
inertial particles and let their four-momenta be vectors $A$, $B$
and $C$ according to the distinguished observer as in the left-hand
side of Fig.\ref{bev}.  Let us note that $A+B=C$, particle $c$ is
obtained by ``fusion'' of particles $a$ and $b$, and particle $a$ is
FTL.  Then particles $a$, $b$ and $c$ form a possible collision
according to the distinguished observer. In the worldview of a new
observer, particles $a$ and $c$ are obtained by ``decay'' of particle
$b$, see the middle of Fig.\ref{bev}.  One of the main axioms of
special relativistic dynamics \ax{SRDyn} is that possible collisions
do not depend on the observer, i.e., relativistic mass and
  linear momentum have to be conserved according to all
  observers. Thus the
four-momenta $A^\new$, $B^\new$ and $C^\new$
of particles $a$, $b$ and $c$ according to the new observer
have to be such that 
\begin{equation}
\label{bev-e}
B^\new=A^\new+C^\new
\end{equation}
 since $a$ and $c$ are obtained by ``decay'' of $b$. Let us try to
define $A^\new$, $B^\new$ and $C^\new$ as the images
of $A$, $B$ and $C$ by the linear part of the Poincar\'e transformation
corresponding to the new observer.         
So let $A'$, $B'$
 and $C'$ be the images of $A$, $B$ and $C$ by the linear part of the
 Poincar\'e transformation. Therefore, $A'+B'=C'$ since $A+B=C$
and $A',B',C'$ are obtained by using a linear transformation, see the 
middle of Fig.\ref{bev}.  Since $B'\neq A'+C'$, equation
   (\ref{bev-e}) does not hold automatically. Thus, if  $A^{\new}$,
 $B^{\new}$ and $C^{\new}$ are  $-A'$, $B'$ and
 $C'$, equation (\ref{bev-e}) is satisfied, see the right-hand
 side of Fig.\ref{bev}.  

This gives the idea to define the
 four-momentum $P^{\new}$ of an arbitrary inertial particle 
 $p$ according to the new observer the following way. Let $P$ be the
 four-momentum of $p$ according to the distinguished observer and let
 $P'$ be the image of $P$ by the linear part of the chosen Poincar\'e
 transformation. $P^{\new}$ is defined to be $P'$ if the time
 component of $P'$ is positive and $P^{\new}$ is defined to be $-P'$
 if the time component of $P'$ is negative (and undefined
 otherwise). This is basically the ``switching-reinterpretation''
 principle  used in \cite{BiDeSu62,Rec86,recami-ftl,Su70,Arntzenius}.

It can be seen that possible collisions do not
 depend on the observer, and relativistic masses
 remain positive. It remains to check that all
 the other axioms of \ax{SRDyn} hold in our
 model. For example, Einstein's first
 postulate, the principle of relativity holds
 basically because the worldviews of
 all the observers are ``alike.''  For a precise proof,
 see\ p.\pageref{proof}.
 
\section{The language of our axiom system}
\label{sec-lang}

To make the
  informal assumptions listed in Section~\ref{sec-inf} precise, we need a
  formal language containing a set of basic symbols for the theory,
i.e., what objects and relations between them we use as basic
concepts.

Here we  use the following two-sorted\footnote{That our theory is
  two-sorted means only that there are two types of basic objects
  (bodies and quantities) as opposed to, e.g., Zermelo--Fraenkel set
  theory where there is only one type of basic objects (sets).}
language of first-order logic parameterized by a natural number $d\ge
2$ representing the dimension of spacetime:
\begin{equation*}
\{\, \B,\Q\,; \IOb, \Ph, +,\cdot,<, \W,\M \,\},
\end{equation*}
where $\B$ (bodies) and $\Q$ (quantities) are the two sorts, $\IOb$
(inertial observers) and $\Ph$ (light signals or photons)
are one-place relation symbols  of sort $\B$, $+$
and $\cdot$ are two-place function symbols and $<$ is a two-place relation
symbol of sort $\Q$,   $\W$ (the worldview
relation) is a $d+2$-place relation symbol the first two arguments of
which are of sort $\B$ and the rest are of sort $\Q$, 
$\M$ (the mass relation) is a $3$-place relation symbol the first two 
arguments of which are of sort $\B$ and the third argument is of sort \Q,
see Fig.\ref{lang}.

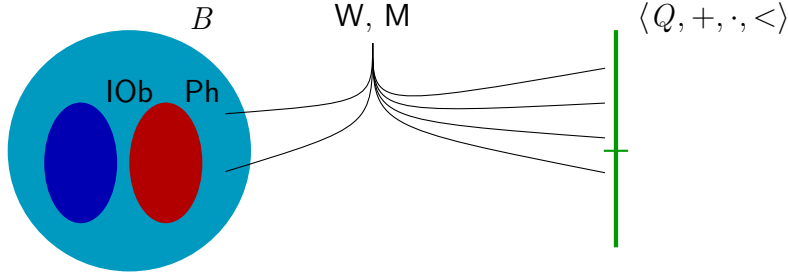
\begin{figure}
\begin{center}
\begin{tikzpicture}[scale=1.6]
\node (B) at (-1.4,1.1) {$\B$};
\node (Q) at (2.8,1.1) {$\langle \Q, +,\cdot,< \rangle$}; 
\node (1) at (-1.3,0.3) {}; 
\node (2) at (-1.3,-0.2) {};
\node (3) at (2.0,-0.2) {};
\node (4) at (2.0,0.1) {};
\node (5) at (2.0,0.4) {};
\node (6) at (2.0,0.7) {};
\fill[lobcolor] (-2.0,0) circle (1);
\fill[lphcolor] (-1.7,-0.1) ellipse  (0.3 and 0.5);
\fill[liobcolor] (-2.4,-0.1) ellipse (0.3 and 0.5);
\node at (-1.4,0.5) {$\Ph$};
\node at (-2.0,0.5) {$\IOb$};
\draw[lqcolor, ultra thick] (2,-0.8) -- (2,1);
\draw[lqcolor, thick] (1.9,0) -- (2.1,0);
\node (W) at (0,1.1) {$\W$,\ $\M$};
\draw (W.south) .. controls (0,0.4)  ..  (1);
\draw (W.south) .. controls (0,0.2)  ..  (2);
\draw (W.south) .. controls (0,0.2)  ..  (3);
\draw (W.south) .. controls (0,0.25)  ..  (4);
\draw (W.south) .. controls (0,0.3)  ..  (5); 
\draw (W.south) .. controls (0,0.35)  ..  (6);
\end{tikzpicture}
\caption{Illustration for the language}
\label{lang}
\end{center}
\end{figure}

Relations $\IOb(k)$ and $\Ph(p)$ are translated as ``\textit{$k$ is
  an inertial observer},'' and ``\textit{$p$ is a light signal or a photon},'' respectively. 
To speak about coordinatization, we
translate $\W(k,b,x_1,x_2,\ldots,x_d)$ as ``\textit{body (observer) $k$
  coordinatizes body $b$ at spacetime location $\langle x_1,
  x_2,\ldots,x_d\rangle$},'' (i.e., at space location $\langle
x_2,\ldots,x_d\rangle$ and instant $x_1$).
Finally we use the mass relation to 
talk about the relativistic masses of
bodies according to inertial observers by reading $\M(k,b,q)$ as
``\textit{the mass of body $b$ is $q$ according to body (observer) $k$.}'' 

{\bf Quantity  terms} are the variables of sort $\Q$ and what can be
built from them by using the two-place operations $+$ and $\cdot$,
{\bf body terms} are only the variables of sort $\B$.
$\IOb(k)$, $\Ph(p)$, $\W(k,b,x_1,\ldots,x_d)$, $\M(k,b,x)$,
$x=y$ and $x< y$  
where $k$, $p$, $b$, $x$, $y$, $x_1$, \ldots, $x_d$ are arbitrary
terms of the respective sorts are so-called {\bf atomic formulas} of
our first-order logic language. The {\bf formulas} are built up from
these atomic formulas by using the logical connectives \textit{not}
($\lnot$), \textit{and} ($\land$), \textit{or} ($\lor$),
\textit{implies} ($\rightarrow$), \textit{if-and-only-if}
($\leftrightarrow$) and the quantifiers \textit{exists} ($\exists$)
and \textit{for all} ($\forall$).

We use the notation $\Q^n$ for the set of all $n$-tuples of elements
of $\Q$. If $\vx\in \Q^n$, we assume that $\vx=\langle
x_1,\ldots,x_n\rangle$, i.e., $x_i$ denotes the
$i$-th component of the $n$-tuple $\vx$. Specially, we write $\W(k,b,\vx)$ in
place of $\W(k,b,x_1,\dots,x_d)$, and we write $\forall \vx$ in place
of $\forall x_1\dots\forall x_d$, etc.

The {\bf models} of
this language are of the form
\begin{equation*}
{\mathfrak{M}} = \langle \B, \Q;
\IOb_\mathfrak{M},\Ph_\mathfrak{M},+_\mathfrak{M},
\cdot_\mathfrak{M},< _{\mathfrak{M}}, \W_\mathfrak{M},\M_\mathfrak{M}\rangle,
\end{equation*}
where $\B$ and $\Q$ are nonempty sets, $\IOb_\mathfrak{M}$ 
and  $\Ph_\mathfrak{M}$ are  unary
relations on $\B$,  $+_\mathfrak{M}$   and $\cdot_\mathfrak{M}$ are binary 
operations and  $<_\mathfrak{M}$ is a binary relation on $\Q$, 
$\W_\mathfrak{M}$
is a subset of $\B\times \B\times \Q^d$ and $\M_{\mathfrak{M}}$ is
a subset of $\B\times\B\times \Q$.  Formulas are interpreted
in $\mathfrak{M}$ in the usual way.  For precise definition of the
syntax and semantics of first-order logic, see, e.g., \cite[\S
  1.3]{CK}, \cite[\S 2.1, \S 2.2]{End}.

We denote that formula $\varphi$ is {\bf valid} in model
$\mathfrak{M}$ by $\mathfrak{M}\models\varphi$.  Formula $\varphi$ 
is {\bf logically implied} by set of formulas, in symbols
$\Sigma\models\varphi$, if{}f (if and only if)  $\varphi$ is
valid in every model of $\Sigma$.

To make our axioms and definitions 
easier to read, we usually omit the outermost universal
quantifiers from our axioms and sometimes we omit them
from the definitions, too, i.e., all the free 
variables are universally quantified.

\section{Axioms for kinematics}
\label{sec-ax}

Here we axiomatize the kinematics of special relativity in our first-order
logic language of Section~\ref{sec-lang}.  Einstein has assumed two postulates
in his 1905 paper \cite{einstein}, the principle of relativity and the light
postulate. The principle of relativity roughly states that the same laws of
nature are true for all inertial observers. Specially they are
indistinguishable from each other by (local) physical experiments, see, e.g.,
Friedman \cite[\S 5]{friedman}.

To formalize the principle of relativity let $\mathcal{P}$ be the  set of formulas of our language with at most one free variable of sort $\B$. Elements of $\mathcal{P}$  play the role of potential ``laws of physics'' in the formulation of the
principle of relativity theory. 
The free variable of sort $\B$ is used
to evaluate these formulas on inertial observers and to check whether they are
valid or not according to the observer in question.  
Now we can formulate the
{\bf strong principle of relativity} 
as the following axiom schema:
\begin{description}
\item[\underline{\ax{SPR^+}}] 
\label{SPR}  
Every potential law of nature
 $\varphi\in \mathcal{P}$ is either true for all the inertial
  observers or false for all of them:
\begin{equation*}
\big\{\, \IOb(k)\land\IOb(h)\rightarrow
\big[\varphi(k,\vx)\leftrightarrow \varphi(h,\vx)\big]
\::\:\varphi\in\mathcal{P}\,\big\}.
\end{equation*}
\end{description}

$\mathcal{P}$
  contains  formulas which may not counted as
laws of nature. Therefore, \ax{SPR^+} 
may be stronger than Einstein's Principle of
Relativity. However, this fact does not concern us now because
  we show here that something does not follow from special relativity,
  and if something does not follow if we use the possibly stronger
  assumption \ax{SPR^+} it does not follow if we use Einstein's
  principle. Let us note here that the difficulty of formulating Einsteins
  principle precisely comes from the fact that the notion of ``laws of
  nature'' is not well-defined.

Einstein assumed without postulating it explicitly that the structure
of quantities is the field of real numbers.  We make this postulate
more general by assuming only the most important algebraic properties
of real numbers for the quantities. 

\begin{description}
\label{ax-fd}
\item[\underline{\ax{AxEField}}] 
 The quantity part $\langle \Q,+,\cdot,< \rangle$ 
is a Euclidean  field, i.e., it is a linearly ordered  
field in the sense of abstract algebra; and 
every positive element has a square root, i.e., $\forall x\,\exists y\,
(x=y^2\lor -x=y^2)$. 
\end{description}

Throughout the paper we assume \ax{AxEField} in our definitions
and axioms without mentioning this explicitly.
We use the usual field operations $0$, $1$, $-$, $/$, and $\gyok$ definable
from $+$ and $\cdot$
within first-order logic. We also use the usual vector-space structure of
$\Q^n$, that is if $\vx,\vy\in\Q^n$ and $q\in\Q$, then $\vx+\vy\in\Q^n$
and $q\cdot\vx\in\Q^n$.

The second postulate of Einstein states that ``Any ray of light
moves in the stationary system of co-ordinates with the determined
velocity $c$, whether the ray be emitted by a stationary or by a
moving body,'' see \cite{einstein}. We can easily formulate this
statement in our first-order logic frame. To do so, let us introduce
the following two concepts.  The {\bf time difference} of coordinate 
points $\vx,\vy\in\Q^d$ is defined as:
\begin{equation*}
\timed(\vx,\vy)\de |x_1-y_1|. 
\end{equation*}
The {\bf spatial distance} of $\vx,\vy\in\Q^d$ is defined as:
\begin{equation*}
  \sqspace(\vx,\vy)\de \sqrt{(x_2-y_2)^2+\ldots+(x_d-y_d)^2}.
\end{equation*}

\begin{description}
\item[\underline{\ax{AxLight}}] 
\label{axlight}
There is \emph{at least one} inertial observer,
  according to whom, any light signal moves with the same speed
  $c$ (independently of the fact that which body emitted the signal).
  Furthermore, it is possible to send out a light signal in any
  direction everywhere
(see Fig.\ref{d13}):
\begin{multline}
\exists kc\Big[\IOb(k)\land 0<c\land
\ \forall\vx\vy\,  
\Big(\exists p \big[ \Ph(p)\land \W(k,p,\vx) \land
\W(k,p,\vy)\big] \leftrightarrow\\ \sqspace(\vx,\vy)=
c\cdot\timed(\vx,\vy)\Big)\Big].
\label{light}
\end{multline}
\end{description}
Axiom \ax{AxLight}, as Einstein's original second postulate, requires
  only the existence of \emph{at least one} inertial observer according to whom all
  light signals move with the same speed. However, by the principle of
  relativity, \ax{AxLight} implies that all light signals move with the same
  speed according to {\em all} the inertial observers.  More precisely axioms
\ax{SPR^+}, \ax{AxLight} and \ax{AxEField} imply that the speed of light is
the same for {\em every} inertial observer in every direction, i.e., formula
(\ref{light}) holds if we replace ``$\exists kc$'' with ``$\exists c\forall
k$'' in it, see \cite[Prop.1]{FTLconsSR}.

\begin{figure}
\begin{center}
\small
\psfrag{ek}[r][r]{$\exists c\exists k$}
\psfrag{fx}[r][r]{$\vx$}
\psfrag{fy}[r][r]{$\vy$}
\psfrag{time}[l][l]{$\timed(\vx,\vy)$}
\psfrag{space}[lt][lt]{$\sqspace(\vx.\vy)$}
\psfrag{ep}[l][l]{$\exists p\; \Leftrightarrow\; \sqspace=c\cdot\timed$}
\psfrag{Ph}[r][r]{\ax{AxLight}}

\psfrag{Self}[r][r]{\ax{AxSelf}}
\psfrag{wl}[r][r]{$k$}

\psfrag{Ev}[c][c]{\ax{AxEv}}
\psfrag{fx}[r][r]{$\vx$}
\psfrag{ey}[l][l]{$\exists \vy$}
\psfrag{text1}[c][c]{$\ev_k(\vx)=\ev_h(\vy)$}
\psfrag{k}[r][r]{$k$}
\psfrag{h}[r][r]{$h$}

\psfrag*{p}[b][b]{$p$}
\psfrag{e1}[t][t]{$e_1$}
\psfrag{e2}[t][t]{$e_2$}
\psfrag{1e}[l][l]{$e_1$}
\psfrag{2e}[l][l]{$e_2$}
\psfrag{sym}[c][c]{\ax{AxSymD}}
\includegraphics[keepaspectratio,width=0.8\textwidth]{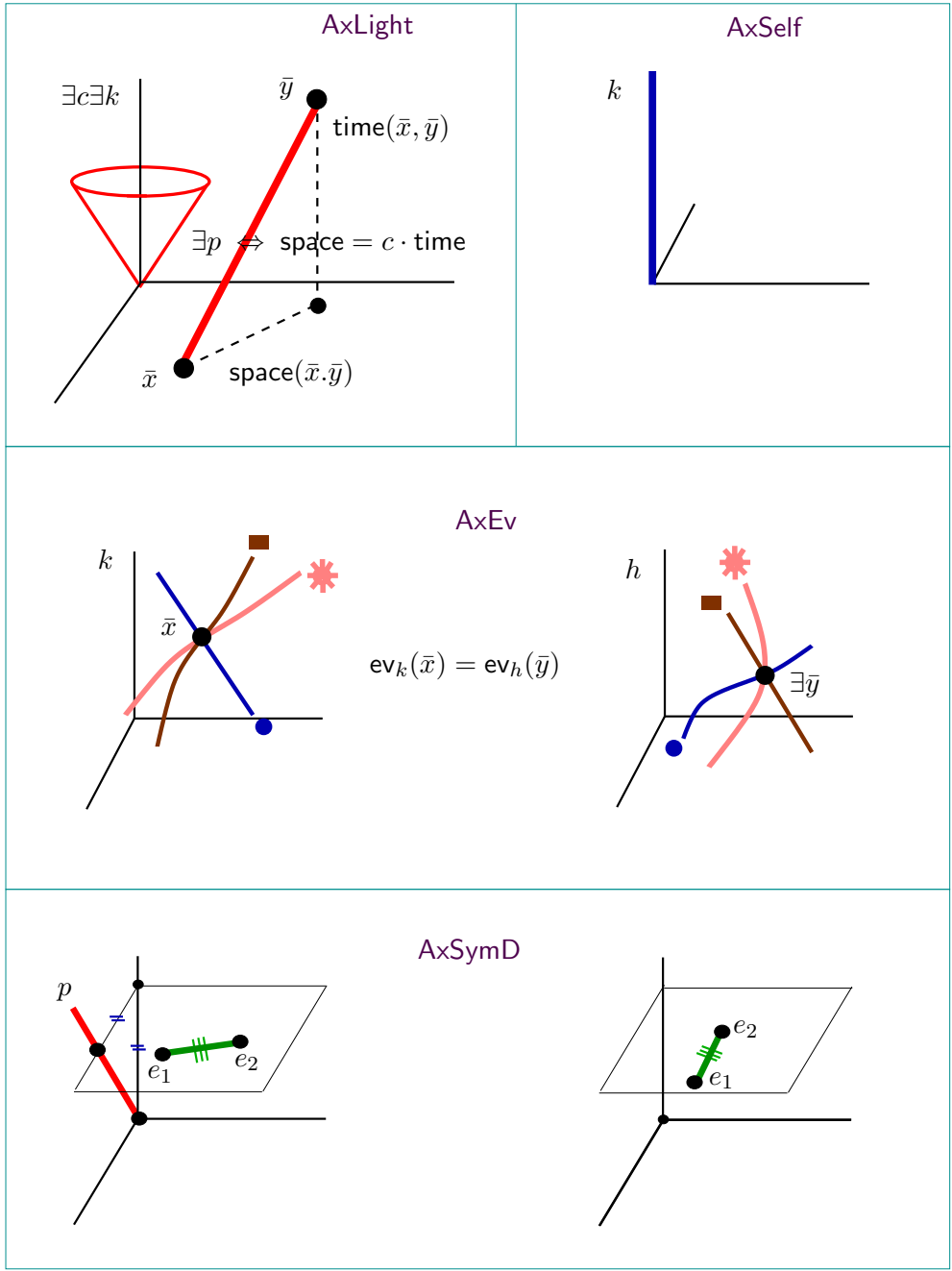}

\caption{Illustration for the axioms of kinematics}
\label{d13}
\end{center}
\end{figure}

As any other approach to relativity theory, we also
  assume that inertial observers coordinatize the same ``external'' reality
(the same set of events).  By the {\bf event} occurring for inertial observer
$k$ at coordinate point $\vx$, we mean the set of bodies $k$
coordinatizes at $\vx$:
\begin{equation*}
\ev_k(\vx)\de\{ b : \W(k,b,\vx)\}.
\end{equation*}

\begin{description}
\item[\underline{\ax{AxEv}}]
\label{AxEv}
All inertial observers coordinatize the same set of events
(see Fig.\ref{d13}):
\begin{equation*}
\IOb(k)\land\IOb(h)\rightarrow \exists \vy\, \forall b
\big[\W(k,b,\vx)\leftrightarrow\W(h,b,\vy)\big].
\end{equation*}
\end{description}

From now on, we  use $\ev_k(\vx)=\ev_h(\vy)$ to abbreviate the
subformula\\ $\forall b\,[\W(k,b,\vx)\leftrightarrow\W(h,b,\vy)]$ of
\ax{AxEv}.

Basically we are ready for formulating the kinematical part of 
Einstein's  special relativity theory within our 
axiomatic framework. Nevertheless, let us
introduce two more simplifying axioms.

\begin{description}
\item[\underline{\ax{AxSelf}}]
\label{axself}
Any inertial observer is stationary according to its own coordinate
system (see Fig.\ref{d13}):
\begin{equation*}
\IOb(k)\rightarrow \forall \vx\big[\W(k,k,\vx) \leftrightarrow
x_2=\ldots=x_d=0\big].
\end{equation*}
\end{description}

Axiom \ax{AxSelf} makes it easier to speak about the motion of
inertial observers since it identifies the observers with their
time-axes. So instead of always referring to the time-axes of
inertial observers we can speak about their motion directly.

Our last axiom on kinematics is a symmetry axiom saying that
all inertial observers use the same units of measurement.
\begin{description}
\item[\underline{\ax{AxSymD}}]
\label{AxSymD}
Any two inertial observers agree as to the spatial distance between
two events if these two events are simultaneous for both of them; and
the speed of light is 1 for all inertial observers (see Fig.\ref{d13}):
\begin{multline*}
\IOb(k)\land\IOb(h) \land x_1=y_1\land
x'_1=y'_1\land \ev_k(\vx)=\ev_h(\vx')\land\ev_k(\vy)=\ev_h(\vy')\\
\rightarrow\quad \sqspace(\vx,\vy)=\sqspace(\vx',\vy'),
\text{ and }\\
\IOb(k)\rightarrow\exists
p\big[\Ph(p)\land\W(k,p,0,\ldots,0)\land\W(k,p,1,1,0,\ldots,0)\big].
\end{multline*}
\end{description}

Axiom \ax{AxSymD}  
simplifies the formulation of
our theorems because we do not have to consider situations such as
when one observer measures distances in meters while another observer
measures them in feet.

Let us now introduce an axiom system \ax{SR} 
for kinematics of special relativity as the
collection of the axioms above:
\begin{equation*}
\ax{SR}  \de  \ax{SPR^+}\cup\{\ax{AxLight},\ax{AxEField},\ax{AxEv},\ax{AxSelf},
\ax{AxSymD}\}.
\end{equation*}

Let us note that usually 
a more general axiom system called \ax{SpecRel} is used 
for axiomatizing special relativity, see,  e.g., \cite{logst,synthese,FTLconsSR}.
\ax{SpecRel}   
 does not contain \ax{SPR^+} but it captures the
    kinematics of special relativity, i.e., it implies that the worldview 
   transformations are Poincar\'e ones. 
  It is proved in \cite{FTLconsSR} that 
  \ax{SpecRel} is more general than \ax{SR}.

To characterize the possible
relations between the worldviews of inertial observers, let us
introduce the {\bf worldview transformation} between observers $k$ and
$h$ (in symbols, $\w_{kh}$) as the binary relation on $\Q^d$
connecting the coordinate points where $k$ and $h$ coordinatize the
same events:
\begin{equation*}\label{eq-ww}
\w_{kh}(\vx,\vy)\defiff\ev_k(\vx)=\ev_h(\vy). 
\end{equation*}

Map $P:\Q^d\rightarrow\Q^d$ is called a {\bf Poincar\'e transformation} if{}f
it is an affine bijection having the following property:
\begin{equation*}
\timed(\vx,\vy)^2-\sqspace(\vx,\vy)^2=\timed(\vx',\vy')^2-
\sqspace(\vx',\vy')^2
\end{equation*}
for all $\vx,\vy,\vx',\vy'\in\Q^d$ for which  $P(\vx)=\vx'$ and $P(\vy)=\vy'$.

\begin{thm}\label{thm-poi}
Let $d\ge3$. Assume \ax{SR}. Then $\w_{kh}$ is a Poincar\'e
transformation if $k$ and $h$ are inertial observers.
\end{thm}
We note that Thm.\ref{thm-poi} also holds, if we replace \ax{SR}
with the  more general axiom system \ax{SpecRel}, see, e.g., 
\cite{MSZRac}.
For versions of Theorem~\ref{thm-poi} using a similar
but different axiom systems of special relativity, see, e.g.,
\cite{pezsgo,AMNSamples,logst}.

Let $\FTL(k,b)$ be the following formula saying that
body $b$ moves FTL according to inertial observer $k$:
\begin{equation}
\label{FTL}
\FTL(k,b)\ \defiff\
\IOb(k)\land \exists \vx\vy[\W(k,b,\vx)\land\W(k,b,\vy)\land
\timed(\vx,\vy)<\sqspace(\vx,\vy)].
\end{equation}

Let \ax{\exists\FTL\IOb} be the following formula saying that there is
an FTL inertial observer:
\begin{equation*}\ax{\exists\FTL\IOb}\ \defiff\ \exists kh\  
[\IOb(h)\land \FTL(k,h)].
\end{equation*}

\noindent
By Thm.\ref{thm-poi}, 
\ax{SR} implies that there are no FTL inertial observers:

\begin{cor}
\label{noftlobsthm} 
Assume $d\geq 3$. Then $\ax{SR}\models\neg\ax{\exists\FTL \IOb}$.
\end{cor} 

We note that, by Thm.\ref{fotetel} on p.\pageref{fotetel},
\ax{SR} does not imply that there are no FTL inertial particles.

We  need the following concepts of kinematics in our axioms for
dynamics.  The {\bf world-line} of body $b$ according to observer $k$
is defined as:
\begin{equation*}
\wl_k(b)\de\{ \vx: \W(k,b,\vx)\}.
\end{equation*}

Body $b$ is called {\bf inertial}
 if for every inertial observer
the world-line of body $b$ is at least two element subset of a straight-line,
formally:
\begin{multline*}
\IOb(k)\ \rightarrow\ \exists \vx\vy\ \Big[\vx\neq\vy\land
\W(k,b,\vx)\land\W(k,b,\vy)\land\\
\big(\W(k,b,\vz)\rightarrow \exists q\ [\Q(q)\land \vz= 
\vx+q\cdot(\vx-\vy)]\big)\Big].
\end{multline*}
The {\bf velocity} $\vel_k(b)$ and the 
{\bf speed} $\vv_k(b)$ of 
inertial body $b$ according to inertial
observer $k$ are defined as follows. Let  
$\vx,\vy$ be such that $\W(k,b,\vx)$, $\W(k,b,\vy)$ and $x_1\neq y_1$. Then 
\begin{equation*}
\label{speed-def}
\vel_k(b)  \leteq  \frac{\langle x_2-y_2,\ldots, x_d-y_d\rangle}{x_1-y_1}
\quad\text{ and }\quad
\vv_k(b)  \leteq  \frac{\sqspace(\vx,\vy)}{\timed(\vx,\vy)},
\end{equation*}
and if there are no such $\vx$ and $\vy$, then $\vel_k(b)$ and
$\vv_k(b)$ are undefined.
For {\it inertial} bodies these are well defined concepts since they do 
not depend on the choice of $\vx$ and $\vy$. $\vv_k(b)<\infty$ abbreviates
that $\vv_k(b)$ is defined, i.e., $\vv_k(b) < \infty$ if{}f $\exists\vx\vy\ 
[\W(k,b,\vx)\land\W(k,b,\vy) \land x_1\neq y_1]$.
We  say that the speed of inertial body $b$ according to
observer $k$ is finite if{}f $\vv_k(b)<\infty$.

\section{Axioms for dynamics}
\label{sec-dinax} 

In this section, we introduce  axioms for dynamics of 
special relativity,  which  are some natural assumptions on 
collisions of inertial particles
and they concern FTL particles, too. 

To introduce the notion of collisions of particles we need some
definitions.  The relativistic mass of body $b$ according to inertial
observer $k$, in symbols $\m_k(b)$, is defined to be $q$ if
$\M(k,b,q)$ holds and there is only one such $q\in\Q$; otherwise
$\m_k(b)$ is undefined. Here we are interested in inertial bodies
having relativistic masses.  Body $b$ is called {\bf inertial
  particle}, in symbols $\IB (b)$, if{}f  $b$ is an inertial body and 
$\m_k(b)$ is defined 
for every inertial observer $k$.

Body $b$ is {\bf incoming} 
{\bf (outgoing)} at coordinate point $\vx$ according to inertial observer $k$,
in symbols $\inc_k(b,\vx)$ ($\out_k(b,\vx)$), 
if{}f $b$ is an inertial particle,
$\vx$ is on the world-line of $b$, and the time component 
of each coordinate point on the world-line of $b$ different from $\bar x$
is less than 
(greater than) the time component of $\vx$ (see the left-hand side of
Fig.\ref{ppcoll} and Fig.\ref{ppcoll1}): 
\begin{eqnarray*}
\inc_k(b,\vx) & \defiff & \IB(b)\lland 
\W(k,b,\vx)\lland \big(\W(k,b,\vy)\rightarrow  [\vy=\vx\llor y_1< x_1]\big)
\label{inc-def},\\
\out_k(b,\vx) & \defiff & \IB(b)\lland
\W(k,b,\vx)\lland
\label{out-def}
\big(\W(k,b,\vy)\rightarrow
[\vy=\vx\llor x_1< y_1]\big).
\end{eqnarray*}

\begin{figure}
\small
\psfrag*{inc}[t][t]{$\inc_k(b_i,\vx)$ (incoming)}
\psfrag*{out}[b][b]{$\out_k(b_i,\vx)$ (outgoing)}
\psfrag*{b1}[r][r]{$b_1$}
\psfrag*{b2}[r][r]{$b_2$}
\psfrag*{b3}[r][r]{$b_3$}
\psfrag*{b4}[l][l]{$b_4$}
\psfrag*{b5}[l][l]{$b_5$}
\psfrag*{pcoll}[l][l]{$\poscoll_k(b_1\dots b_5)\; \Leftrightarrow\; A+B+C=D+E$}
\psfrag*{A}[l][l]{$A$}
\psfrag*{B}[l][l]{$B$}
\psfrag*{C}[l][l]{$C$}
\psfrag*{D}[r][r]{$D$}
\psfrag*{E}[r][r]{$E$}
\psfrag*{x}[l][l]{$\vx$}
\psfrag*{a}[r][r]{$a$}
\psfrag*{b}[l][l]{$b$}
\psfrag*{c}[l][l]{$c$}
\psfrag*{a1}[l][l]{$A$}
\psfrag*{bb}[l][l]{$B$}
\psfrag*{c1}[r][r]{$C$}
\psfrag*{inecoll}[l][l]{$\vx\mbox{-}\inecoll_k(ab)$}
\psfrag*{k}[r][r]{$k$}

\includegraphics[keepaspectratio,width=\textwidth]{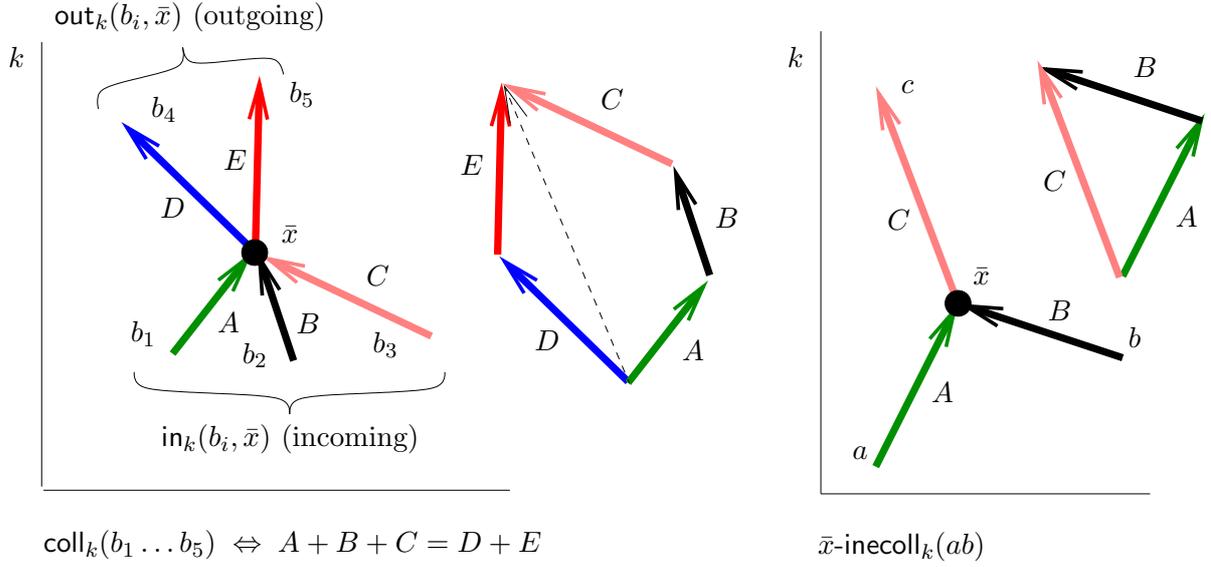}
\caption{Illustration for incoming, outgoing, possible collision 
and inelastic collision of bodies; the vectors ($A,B,C,D,E\in\Q^d$) in 
the figure are the four-momenta of inertial particles, i.e., 
$\langle \m_k(b_i),\m_k(b_i)\cdot\vel_k(b_i)\rangle$, cf.\ (\ref {fourm-d})
and Fig.\ref{fmom}}
\label{ppcoll}
\end{figure}

Let us define the  {\bf possible collisions} of 
bodies as follows. 
Bodies $b_1,\ldots, b_n$ form a possible collision according to observer $k$
if there is a 
coordinate point such that all the bodies are incoming or outgoing
in that coordinate point and  the sum of the relativistic masses of the 
incoming bodies coincides with that of the outgoing ones, and the same holds 
for the linear momenta of the bodies (see Fig.\ref{ppcoll}):
\begin{multline}
\label{pcoll-e} 
\poscoll_k(b_1\ldots b_n)\ \defiff\ 
\exists\vx \Big[\bigwedge_{i=1}^{n} [\inc_k(b_i,\vx)\vee\out_k(b_i,\vx)]\ \land\\
\sum_{\{i\, :\, \inc_k(b_i,\vx)\}}
\m_k(b_i)=\sum_{\{i\, :\,\out_k(b_i,\vx)\}} \m_k(b_i)
\ \land\\ 
 \sum_{\{i\, :\, \inc_k(b_i,\vx)\}}  \m_k(b_i)\cdot \vel_k(b_i)=
\sum_{\{i\, :\,\out_k(b_i,\vx)\}} \m_k(b_i)\cdot \vel_k(b_i)\Big].
\end{multline}

Let us note that, if bodies $b_1,\ldots,b_n$ form a possible collision, then
they are inertial particles by the definition of incoming and outgoing
particles.

For every natural number $n$ we introduce 
an axiom 
saying that possible collisions formed by
$n$ bodies do not depend on the inertial observer. Thus conservations of relativistic
mass and linear momentum do not depend on the inertial observer.

\begin{description}
\item[\underline{\ax{AxColl_n}}]
\label{coll}
 If bodies $b_1,\dots, b_n$ form a 
possible collision 
for an inertial observer,  they form a possible collision
for every inertial observer according to whom the speed of each of them
is finite (see Fig.\ref{d12}):  

\begin{equation*}
\IOb(k)\land\IOb(h)\land\bigwedge_{i=1}^{n}\sqspeed_h(b_i)<\infty\land
\poscoll_k(b_1\ldots b_n)\ \rightarrow\ \poscoll_h(b_1\ldots b_n).
\end{equation*}
\end{description}

\noindent
Let \ax{Coll} be the axiom schema containing \ax{AxColl_n} for every
natural number $n$:

\begin{description} 
\item[\underline{\ax{Coll}}] 
\label{axcoll}
Possible collisions do not depend on the inertial observer:
\begin{equation*}
\ax{Coll}\leteq\{\,\ax{AxColl_n}\::\:n \text{ is a natural number}\,\}. 
\end{equation*}
\end{description}

Bodies $a$ and $b$ {\bf collide inelastically}  
according to inertial observer $k$ at coordinate point $\vx$, 
in symbols $\vx\text{-}\inecoll_{k}(ab)$, if{}f there is a body $c$ such
that $a,b,c$ form a possible collision, $a,b$ are incoming and $c$ is  
outgoing at $\vx$ (see the right-hand side of 
Fig.\ref{ppcoll}):
\begin{equation}
\label{xinecoll-d}
\vx\text{-}\inecoll_k(ab)\ \defiff\
 \exists c\ [\poscoll_k(abc)\land
\inc_k(a,\vx)\land\inc_k(b,\vx)\land
\out_k(c,\vx)].
\end{equation}

\begin{figure}
\begin{center}
\small
\psfrag*{k}[r][r]{$k$}
\psfrag*{i}[r][r]{$\forall k$}
\psfrag*{k}[r][r]{$k$}
\psfrag*{fa}[r][r]{$\forall a$}
\psfrag*{fb}[l][l]{$\forall b$}
\psfrag*{fp}[l][l]{$\forall\vx$}
\psfrag*{ea}[r][r]{$\exists\; {a'}$}
\psfrag*{eb}[l][l]{$\exists\; {b'}$}
\psfrag*{p}[lb][lb]{$\vx$}
\psfrag*{text}[cb][cb]{$\ldots\m_k(a)+\m_k(b)\neq 0\ \rightarrow\ 
\exists  a'b'\ \vx\text{-}\inecoll(a'b')\ldots$}

\includegraphics[keepaspectratio, width=0.9\textwidth]{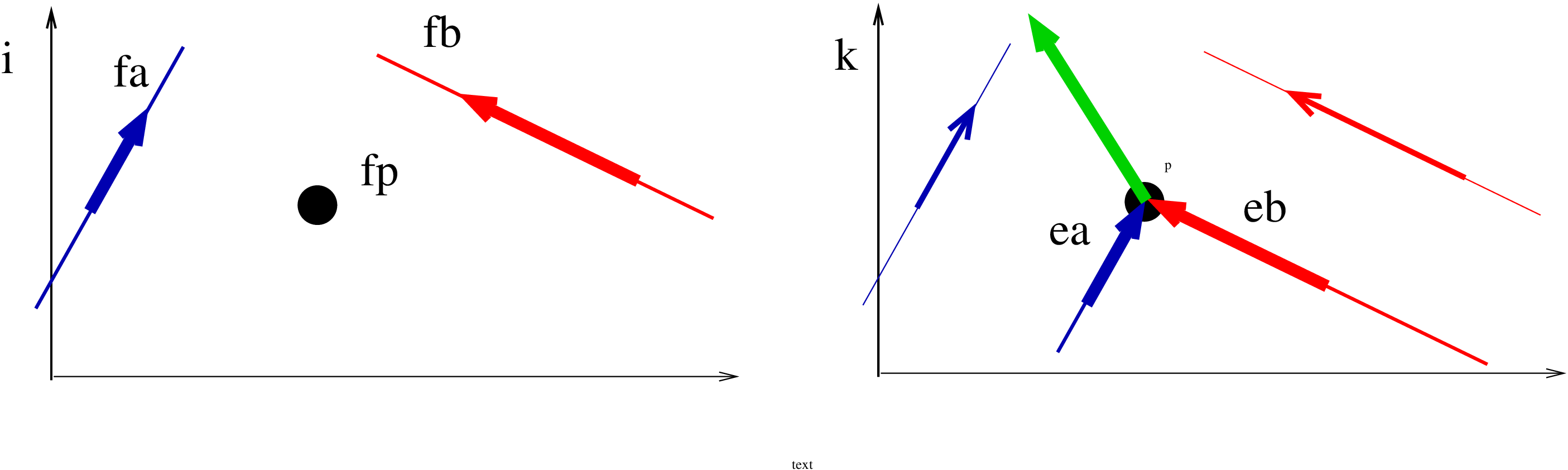}
\caption{Illustration for \ax{Ax\forall\inecoll}}
\label{inec3}
\end{center}
\end{figure}

By the following axiom,  
particles can be collided inelastically at any coordinate point.

\begin{description}
\item[\underline{\ax{Ax\forall \inecoll}}] 
\label{axinecoll}
For every inertial observer, every coordinate
point and every two inertial particles $a$ and $b$, if the sum of
their relativistic masses is nonzero and their speeds are finite,
there are inertial particles $a'$ and $b'$ such that they collide
inelastically at the given coordinate point and the relativistic
masses and velocities of $a'$ and $b'$ coincide with those of $a$ and
$b$, respectively (see Fig.\ref{inec3} and Fig.\ref{d12}):
\begin{multline*}
\IOb(k)\land\Ip(a)\land\Ip(b)\land
\sqspeed_k(a)<\infty\land\sqspeed_k(b)<\infty\land
\m_k(a)+\m_k(b)\neq 0\ \rightarrow\\ \exists a'b'\ 
[\vx\text{-}\inecoll_k(a'b')\land\\
\m_k(a')=\m_k(a)\land
\m_k(b')=\m_k(b)\land\vel_k(a')=\vel_k(a)\land\vel_k(b')=\vel_k(b)].
\end{multline*}
\end{description}

We assume that relativistic
masses of slower than light inertial particles depend only on their
speeds.

\begin{description}  
\item[\underline{\ax{AxSpd}}] 
\label{axspeed}
If an inertial
particle is moving with the same slower than light speed
  according to two inertial observers, then the relativistic masses of the
particle are the same for them (see Fig.\ref{d12}):
\begin{equation*}
\IOb(k)\land\IOb(h)\land\IB(b)\land
\vv_k(b)=\vv_h(b)<1\ 
\rightarrow\  \m_k(b) = \m_h(b).
\end{equation*}
\end{description}

We also assume that, if two inertial particles
  have the same velocities and relativistic masses according to an
  inertial observer, then they have the same relativistic masses
  according to every inertial observer.
\begin{description}
\item[\underline{\ax{AxMass}}]
\label{axmass} 
If the relativistic masses and velocities of two inertial 
particles  coincide
for an inertial observer, then their relativistic masses coincide for 
every inertial observer 
(see Fig.\ref{d12}): 
\begin{equation*}
\IOb(k)\land\IOb(h)\land\IB(a)\land\IB(b)\land
\m_k(a)=\m_k(b)\land\vel_k(a)=\vel_k(b)\ \rightarrow\  \m_h(a)=\m_h(b).
\end{equation*}
\end{description}

\begin{figure}[h!btp]
\small
\psfrag{k}[r][r]{$k$}
\psfrag{h}[r][r]{$h$}
\psfrag{b}[r][r]{$b$}
\psfrag{lb}[l][l]{$b$}
\psfrag{lc}[l][l]{$a$}
\psfrag*{ea}[r][r]{$\exists a'$}
\psfrag*{eb}[l][l]{$\exists b'$}

\psfrag{AxMass}[lt][lt]{\ax{AxMass}\quad
$\m_k(a)=\m_k(b)\rightarrow\m_h(a)=\m_h(b)$}

\psfrag{AxDyn}[lt][lt]{\ax{AxSpd}\quad 
$\vv_k(b)=\vv_h(b)\rightarrow\m_k(b)=\m_h(b)$}

\psfrag{AxTrColl}[lt][lt]
{\ax{Coll}, \ax{AxColl_4}\quad $\poscoll_k(abcd)\rightarrow 
\poscoll_h(abcd)$}
\psfrag{text}[lt][lt]{$\exists a'b'\ \vx\text{-}\inecoll_{k}(a'b')$} 

\psfrag{ma}[r][r]{$a$}
\psfrag{mb}[l][l]{$b$}
\psfrag{mx}[l][l]{$\vx$}
\psfrag{C}[b][b]{$c$}
\psfrag{D}[b][b]{$d$}
\psfrag{A}[t][t]{$a$}
\psfrag{B}[t][t]{$b$}
\psfrag{tC}[t][t]{$c$}
\psfrag{Axinecoll}[lt][lt]{\ax{Ax\forall inecoll}}
\includegraphics[keepaspectratio,width=\textwidth]{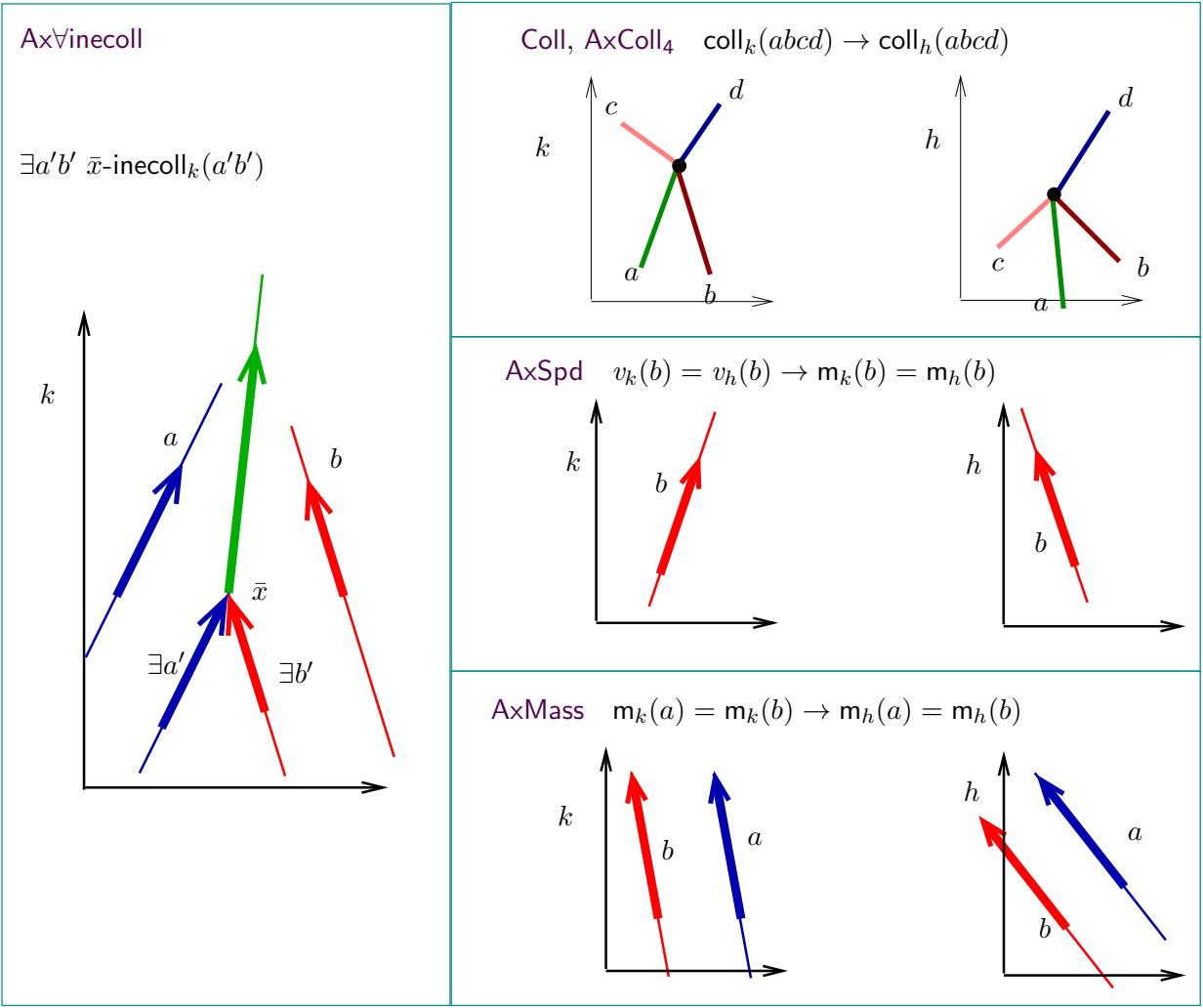}
\caption{Illustration for axioms of dynamics}
\label{d12}
\end{figure}
 
To avoid trivial models, we also assume that there are inertial
observers moving relative to each other and there are inertial
particles of arbitrary positive relativistic masses and arbitrary
non-FTL velocities. 


\begin{description}
\item[\underline{\ax{AxThEx^+}}] 
\label{axthexp}
Inertial observers can move along any
straight line of  slower than light speed 
and inertial particles
of arbitrary positive relativistic masses  can move along any
straight line of non-FTL speed:
\begin{equation*}
\big(\IOb(k)\land \sqspace(\vx,\vy)<\timed(\vx,\vy)\  \rightarrow\ 
\exists h\ [\IOb(h)\land\W(k,h,\vx)\land\W(k,h,\vy)]\big)\ \land
\end{equation*}
\begin{multline*}
\big(\IOb(k)\land\sqspace(\vx,\vy)\leq \timed(\vx,\vy)\land 0<q\
\rightarrow\\
\exists b\ [\Ip(b)\land \W(k,b,\vx)\land\W(k,b,\vy)\land
\m_k(b)=q]\big).
\end{multline*}
\end{description}

 By the following axiom, every potential collision can be realized.
\begin{description}
\item[\underline{\ax{Ax\forall Coll}}] 
\label{axminden}
For every inertial observer, coordinate point and inertial
particle $a$ of finite speed, there is an inertial particle $b$ such that
the relativistic mass and velocity of $b$ coincide with those of $a$, and
$b$ is outgoing (incoming)  at the given coordinate point.
\begin{align*}
\IOb(k)\land\Ip(a)\land &  \vv_k(a)<\infty\  \rightarrow\\ 
 \Big( & \exists b\,[\out_k(b,\vx)\land\m_k(b)
=\m_k(a)\land\vel_k(b)=\vel_k(a)]\ \land\\
 & \exists b\,[\ \inc_k(b,\vx)\ \land\m_k(b)
=\m_k(a)\land \vel_k(b)=\vel_k(a)]\Big). 
\end{align*}
\end{description}
By definition of possible collision, \ax{Ax\forall Coll} is   
equivalent with \ax{\forall Coll} below 
saying that, according to every inertial observer, every potential
 collision can be realized at any coordinate point. 
\begin{description}
\item[\ax{\forall Coll}]
For every positive integer $n$,
nonnegative integer $l$ with $l\leq n$, inertial observer $k$, coordinate point
$\vx\in\Q^d$, and (not necessarily different)
inertial particles $a_1,\ldots,a_n$ with 
\begin{equation*}
\sum_{0<i\leq l}\m_k(a_i)=\sum_{l<i\leq n}\m_k(a_i)\ \land\
\sum_{0<i\leq l}\m_k(a_i)\cdot\vel_k(a_i)=\sum_{l<i\leq n}
\m_k(a_i)\cdot\vel_k(a_i),
\end{equation*}
there are inertial particles $b_1,\ldots,b_n$ such that 
$\m_k(b_i)=\m_k(a_i)$ and
 $\vel_k(b_i)=\vel_k(a_i)$ for every $i$, 
$\inc_k(b_i,\vx)$ for every $0<i\leq l$,  $\out_k(b_i,\vx)$ for every
$l<i\leq n$, and therefore $\coll_k(b_1\ldots b_n)$.
\end{description}

Velocity $\vu\in\Q^{d-1}$ is said to be {\bf non-FTL}
if $u_1^2+\ldots+u_{d-1}^2\leq 1$.

\ax{AxThEx^+} and \ax{Ax\forall Coll} imply the following
statement: 
\begin{description}
\item[$\qquad $]
For every positive integer $n$, nonnegative integer $l$ with 
$l\leq n$, inertial observer $k$, coordinate point
$\vx\in\Q^d$, positive $m_1,\ldots,m_n\in\Q$, and non-FTL velocities  
$\vvvv_1,\ldots,\vvvv_n\in\Q^{d-1}$ with
\begin{equation*}
\sum_{0<i\leq l} m_i =\sum_{l<i\leq n} m_i\ \land\
\sum_{0<i\leq l} m_i\cdot \vvvv_i=\sum_{l<i\leq n}
m_i\cdot \vvvv_i
\end{equation*}
there are inertial particles $b_1,\ldots,b_n$ such that 
$\m_k(b_i)=m_i$ and $\vel_k(b_i)=\vvvv_i$ for every $i$, 
$\inc_k(b_i,\vx)$ for every $0<i\leq l$,  $\out_k(b_i,\vx)$ for every
$l<i\leq n$, and therefore $\coll_k(b_1\ldots b_n)$.
\end{description}

Let
us  introduce an axiom system \ax{SRDyn} for dynamics of special
relativity as the collection of all the axioms of kinematics and dynamics
above:
\begin{equation*}
\label{dinamika}
\ax{SRDyn}  \leteq   \ax{SR}\cup \ax{Coll}\cup\{\ax{Ax\forall
inecoll}, \ax{AxSpd}, \ax{AxMass}, \ax{AxThEx^+}, \ax{Ax\forall Coll}\}.
\end{equation*}

\section{Independence of massive FTL inertial particles of \ax{SRDyn}}
\label{sec-main}

Now we show that the existence of massive FTL inertial particles is
independent of \ax{SRDyn}. To formulate this statement we need some
definitions.

Let us recall that formula $\FTL(k,b)$ states that body $b$ moves FTL
according to inertial observer $k$, see (\ref{FTL}) on
p.\pageref{FTL}.

Let $\exists \FTL \Ip$ be the following formula saying that there is
an FTL inertial particle having positive relativistic mass:
\begin{equation*}
\exists\FTL \Ip\ \defiff\ \exists kb\ [\Ip(b)\land\FTL(k,b)\land \m_k(b)>0].
\end{equation*}

Let $\Sigma$ be a set of formulas and $\varphi$ be a formula.
$\Sigma\not\models\varphi$ denotes 
that
$\varphi$ is not  implied by $\Sigma$, i.e., there
is a model of $\Sigma$ in which $\varphi$ is not valid. 
Statement  $\varphi$ is called 
{\bf independent} of $\Sigma$ if neither $\varphi$ nor its negation
$\neg\varphi$ is  implied by $\Sigma$, i.e.,
$\Sigma\not\models\varphi$ and $\Sigma\not\models\neg\varphi$. 
Let us note
that $\varphi$ is independent of $\Sigma$ if there are two models of
$\Sigma$ such that $\varphi$ is valid in one model and $\neg\varphi$
is valid in the other one.

The main result of the present paper is Theorem~\ref{fotetel} below. It says
that the existence of massive FTL inertial particles is independent of relativistic
dynamics.

\begin{thm}
\label{fotetel}  
\ax{\exists\FTL\Ip} is
independent of $\ax{SRDyn}$, that is
\begin{eqnarray*}
\ax{SRDyn}& \not\models &  \ax{\exists\FTL\Ip}, \text{ and}\\
\ax{SRDyn}& \not\models & \neg\ax{\exists\FTL\Ip},
\end{eqnarray*}
equivalently, both \ax{\exists\FTL\Ip} and 
$\neg\ax{\exists\FTL\Ip}$ are consistent with \ax{SRDyn}. 
\end{thm}
\begin{proof}
The theorem is a corollary of Thm.\ref{masodikfotetel} below.
\end{proof}

By Theorem~\ref{masodikfotetel} below, the existence of
massive FTL inertial particles is independent of relativistic dynamics even
  if we assume that the structure of quantities is isomorphic to the
  field of real numbers (or to any other fixed Euclidean field).

\begin{thm}
\label{masodikfotetel}
For every $d\geq 2$ and 
for every Euclidean  field $\mathfrak{Q}$, there are models
$\mathfrak{M}_1$ and $\mathfrak{M}_2$ of
$\ax{SRDyn}$ such that $\mathfrak{M}_1\models\ax{\exists \FTL \Ip}$, 
$\mathfrak{M}_2\models\neg\ax{\exists \FTL \Ip}$,  and 
$\mathfrak{Q}$ is the field reduct of both  models.
\end{thm}

Based on the intuitive idea in 
Section~\ref{sec-proofinf}, we give a formal proof here using the
following concepts.

Let $f:\Q^d\rightarrow\Q^d$ and $g:\Q^d\rightarrow\Q^d$ be maps.
$f\circ g$ denotes the composition of the two maps, i.e.,
$(f\circ g)(\vx)=f\big(g(\vx)\big)$. 
 $f^{-1}$ denotes the inverse map of $f$.
Let $H$ be a subset of $\Q^d$. The {\bf $f$-image} of set $H$ is defined as:
$f[H]\de\{\,f(\vx)\::\:\vx\in H\,\}$.
The {\bf identity map} is defined as:
$\Id(\vx)\de\vx$ for all $\vx\in\Q^d$.
Let $\vx,\vy\in\Q^d$. Then
$\ray{\vx\vy}$ denotes the closed ray (or half-line) with 
initial point $\vx$ and containing $\vy$, i.e.,
$\ray{\vx\vy} \leteq\{ \vx+q\cdot(\vy-\vx)\, :\,
0\leq q\}$.

The {\bf time-axis} is defined as 
\begin{equation*}
\taxis\de\{\,\vx\in\Q^d\::\: x_2=\ldots=x_d=0\,\}. 
\end{equation*}

Let $k\in\IOb$ and $b\in\B$.  The {\bf four-momentum} ${\fvp_k(b)}$ of
body $b$ according to inertial observer $k$ is defined as the element
of $\Q^d$ whose time component is the relativistic mass and
space component is the 
linear momentum of $b$ according to $k$ if $b$ is an
inertial particle and the speed of $b$ is finite (see Fig.\ref{fmom}),
i.e.:
 \begin{equation}
\label{fourm-d}
 \fvp_k(b)_1=\m_k(b)\  \text{and}\ 
\langle \fvp_k(b)_2,\ldots,\fvp_k(b)_d\rangle =\m_k(b)\cdot \vel_k(b),
 \end{equation}
 if $\Ip(b)$ and $\sqspeed_k(b)<\infty$,
and $\fvp_k(b)$ is undefined otherwise.
 It is not difficult to prove that $\fvp_k(b)$ is parallel to the
 world-line of $b$.

\begin{figure}[h!bt]
\begin{center}
\small
\psfrag*{P}[l][l]{$\fvp_{{k}}({b})$ ({\bf four-momentum})}
\psfrag*{b}[l][l]{$\wl_k({b})$}
\psfrag*{m}[r][r]{$\m_k(b)$ (mass)}
\psfrag*{l}[l][l]{$\m_k(b)\cdot\vel_k(b)$ (linear momentum)}
\psfrag*{i}[r][r]{$k$}
\includegraphics[keepaspectratio, width=0.5\textwidth]{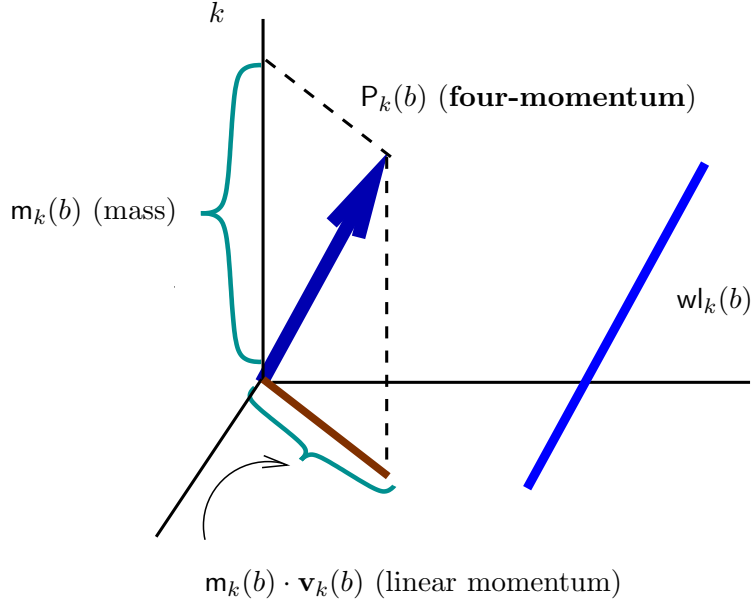}
\caption{Illustration for four-momentum}
\label{fmom}
\end{center}
\end{figure}

Using the concept of four-momentum 
definition (\ref{pcoll-e})  of possible collision of bodies 
can be written in a simpler form: 
\begin{multline}
\label{pcollfm-d}
\poscoll_k(b_1\ldots b_n)\ \Longleftrightarrow\ 
\exists\vx \bigwedge_{i=1}^{n} [\inc_k(b_i,\vx)\vee\out_k(b_i,\vx)]\ \land\\
\sum_{\{i\, :\, \inc_k(b_i,\vx)\}}
\fvp_k(b_i)=\sum_{\{i\, :\,\out_k(b_i,\vx)\}} \fvp_k(b_i). 
\end{multline}

\begin{proof}
\label{proof}
The idea of the proof is in Section~\ref{sec-proofinf} on 
 p.\pageref{sec-proofinf}.

Let $\mathfrak{Q}=\langle \Q;+,\cdot,<\rangle$ be a Euclidean 
field.
We are going to prove our statement by constructing two models
$\mathfrak{M}_1$ and $\mathfrak{M}_2$ of \ax{SRDyn} such that 
in
$\mathfrak{M}_1$ there are massive FTL inertial particles, in 
$\mathfrak{M}_2$ there are no FTL inertial particles and the ordered field
reduct of both models is $\mathfrak{Q}$.  
There is  only a slight difference in the two
constructions. Therefore, we are going to construct the two models
simultaneously. By $\vx\vy$ we denote the ordered pair
$\langle\vx,\vy\rangle$.

\begin{eqnarray}
\label{eq-IOb}
\IOb&  \de & \{\,\text{Poincar\'e transformations of }\Q^d\,\},\\
\label{imp1-d}
\Ip_1  & \de & 
\{\, \vx\vy\in\Q^d\times\Q^d \, :\, \vx\neq\vy\}, \\
\label{imp2-d}
\Ip_2 &\de &\{\, \vx\vy\in\Ip_1 \, :\, \sqspace(\vx,\vy)
\leq\timed(\vx,\vy)\,\}, \text{and}\\
\label{Ls-d}
\Ph &\de & \{\,\vx\vy\in\Ip_1 \, :\, \sqspace(\vx,\vy)=\timed(\vx,\vy)\,\}.
\end{eqnarray}

Let us note that $\Ph\subseteq\Ip_2\subseteq\Ip_1$.
The only difference between the construction of models
$\mathfrak{M}_1$ and $\mathfrak{M}_2$ is in the definition of the set
of bodies:
\begin{equation}
\B_1 \leteq \IOb\cup\Ip_1 \text{ and } \B_2 \leteq \IOb\cup\Ip_2.
\end{equation}
Throughout the proof
$\B$ and $\Ip$ denote $\B_1$ and $\Ip_1$ in the case of
${\mathfrak M}_1$ and denote $\B_2$ and $\Ip_2$ in the
case of ${\mathfrak M}_2$. Furthermore, by ``inertial particles''
we mean the members of $\Ip_1$ in the case of ${\mathfrak M}_1$ and
the members $\Ip_2$ in the case of ${\mathfrak M}_2$.

The model construction is illustrated in Fig.\ref{dyn3} and its
intuitive idea is the following: The worldview transformation between
inertial observers identity $\Id$ and $k$ will be Poincar\'e
transformation $k$. First we define the world-lines of bodies
according to observer $\Id$. In particular, the world-line of particle
$\vx\vy$ is $\ray{\vx\vy}$ for every $\vx\vy$. We transform the
world-lines by transformation $k$ to obtain world-lines according to
arbitrary observer $k$, cf.\ Fig.\ref{dyn3}.
\begin{figure}[h!bt]
\small
\psfrag*{wl}[lb][lb]{$\wl(\vx\vy)$}
\psfrag*{wl1}[lb][lb]{$\wl_k(\vx\vy)$}
\psfrag*{text1}[l][l]{worldview of observer $k$}
\psfrag*{text2}[l][l]{worldview of observer $\Id$}
\psfrag*{t}[lb][lb]{$\taxis$}
\psfrag*{h-}[b][b]{$h^{-1}$}
\psfrag*{wh}[l][l]{$\wl(h)$}
\psfrag*{x}[l][l]{$\vx$}
\psfrag*{y}[l][l]{$\vy$}
\psfrag*{m}[l][l]{$\m(\vx\vy)$}
\psfrag*{k}[b][b]{$k$}
\psfrag*{wkh}[b][b]{$\wl_k(h)$}
\psfrag*{kx}[r][r]{$k(\vx)$}
\psfrag*{ky}[r][r]{$k(\vy)$}
\psfrag*{mk}[b][b]{$\m_k(\vx\vy)$}
\includegraphics[keepaspectratio,width=\textwidth]{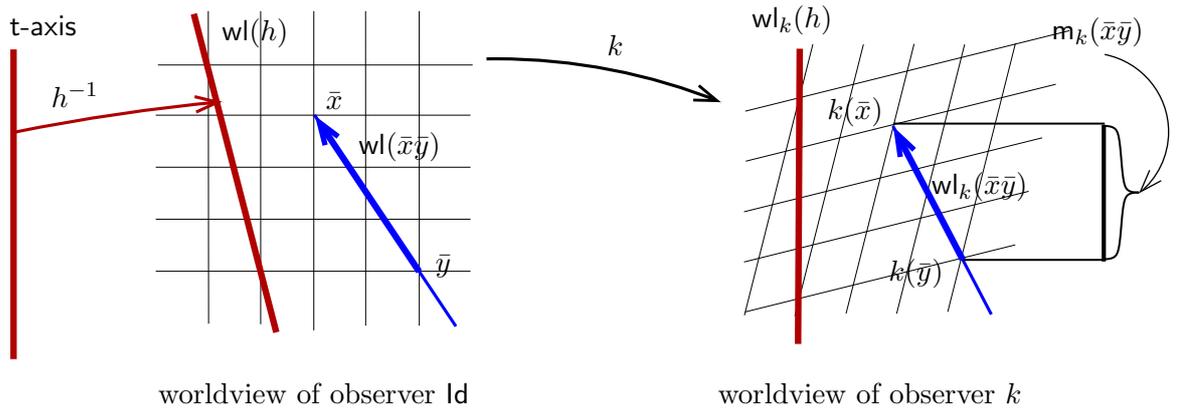}
\caption{Illustration for the construction of the models}
\label{dyn3}
\end{figure}
Thus the world-line of particle
$\vx\vy$ is $\ray{k(\vx)k(\vy)}$ according to observer $k$.

We define relativistic masses such that (i)-(iii) holds. (i) 
The time components of the four-momenta are positive, 
thus the relativistic masses of inertial particles of finite speeds 
are positive.  (ii) The four-momentum $\fvp(\vx\vy)$ of 
particle $\vx\vy$,  according to observer $\Id$, 
is one of the two vectors connecting $\vx$ and $\vy$,
i.e., one of $\vx-\vy$ and $\vy-\vx$, if the speed of 
$\vx\vy$ is finite. (iii)  According to observer $k$, the four-momentum of 
particle $\vx\vy$ with finite speed is one of the two vectors connecting
$k(\vx)$ and $k(\vy)$.

Now we construct our models based on the intuitive idea above.  
See Fig.\ref{dyn3}.
In the worldview of observer $\Id$, we define
the world-line of inertial observer $h$ and the world-line of inertial 
particle $\vx\vy$ as:
\begin{equation}
\label{wl-i}
\wl(h) \leteq  h^{-1}[\taxis],\ \text{ and }
\wl(\vx\vy)  \leteq  {\ray{\vx\vy}},
\end{equation}
see Fig.\ref{dyn3}. We define the world-line of body
$b$ and the relativistic mass of inertial particle $\vx\vy$ in the
worldview of observer $k$ as:

\begin{equation}    
\wl_k(b) \leteq k[\wl(b)]\ \text{ and }\ \label{wl-d}
\m_k(\vx\vy)\leteq  \timed\big(k(\vx),k(\vy)\big),
\end{equation}
cf.\ Fig.\ref{dyn3}.

Finally  we define the worldview relation $\W$ and the mass relation
$\M$ as:
\begin{eqnarray}
\label{W-d}
\W(k,b,\vx) & \defiff & k\in\IOb\land b\in \B \land\vx\in\wl_k(b), \text{ and}\\
\label{M-d}  
\M(k,b,q) & \defiff & k\in\IOb\land b\in\Ip\land \m_k(b)=q.
\end{eqnarray}

Now models ${\mathfrak M}_1$ and ${\mathfrak M}_2$ are given.
It is easy to see that the set of 
inertial particles in ${\mathfrak M}_1$ and ${\mathfrak M}_2$ are
$\Ip_1$ and $\Ip_2$, respectively.

For every inertial observers $k$ and $h$
and inertial particle
$\vx\vy$,  by (\ref{wl-d}), it is easy to see that
\begin{eqnarray}
\label{wlray}
\wl_k(\vx\vy)& = & \ray{k(\vx)k(\vy)}, \\
\label{harmadikcs}
\fvp_k(\vx\vy)   &  = & \left\{\begin{array}{ll} 
k(\vx)-k(\vy)    & \text{ if  } k(\vy)_1<k(\vx)_1, \\
k(\vy)-k(\vx)    & \text{ if } k(\vx)_1< k(\vy)_1,\\
\text{undefined} &\text{  if  }  k(\vx)_1=k(\vy)_1, 
\end{array}
\right.
\\
\label{eeeee}
\sqspeed_k(\vx\vy)   &  = & \left\{\begin{array}{ll} 
\frac{\sqspace(k(\vx),k(\vy))}{\timed{(k(\vx),k(\vy))}}    & \text{ if  }
k(\vy)_1\neq k(\vx)_1, \\
\text{undefined}    & \text{ if } k(\vx)_1=k(\vy)_1, 
\end{array}
\right.
\\ \w_{hk} &  =  &  k\circ h^{-1} \text{ (is a Poincar\'e transformation), and}
\label{roe}
\\ \label{taxis} \wl_k(h) &  = & k\circ h^{-1}[\taxis].
\label{ue}
\end{eqnarray}

By (\ref{pcollfm-d})--(\ref{ue}),  
it is not difficult to prove that
$\mathfrak{M}_1$ and $\mathfrak{M}_2$ are models of
$\ax{SRDyn}\setminus\ax{SPR^+}$, there are FTL
inertial particles in $\mathfrak{M}_1$
and there are no FTL inertial particles in $\mathfrak{M}_2$.
For details of the proof see below. 

By \cite[Prop.2]{FTLconsSR}, to prove that axiom schema \ax{SPR^+} is 
valid in models $\mathfrak{M}_1$ and $\mathfrak{M}_2$, it is enough to show
that, for every inertial observer  $k$, there is an automorphism
fixing the quantities and taking observer $k$ to observer $\Id$. For fixed
observer $k$,  let $\alpha$ be the following map:
\begin{equation}
\label{alpha}  
\alpha(h)= h\circ k^{-1},\ \ \alpha(q)   =   q\ \text{ and }\ 
\alpha(\vx\vy)  =  k(\vx)k(\vy)
\end{equation} 
for every inertial observer $h$, quantity $q$, and
inertial particle $\vx\vy$. Clearly $\alpha$ takes observer $k$ to
observer $\Id$. It is not difficult to prove that $\alpha$ is an
automorphism of our model, for details of the proof see below. 
Thus \ax{SPR^+} is valid in the models. 

\medskip
\hosszu{{\bf Details of the proof:}
Now we are going to show in detail that $\mathfrak{M}_1$  and $\mathfrak{M}_2$
are models of axiom system \ax{SRDyn} and there are massive FTL inertial particles 
in $\mathfrak{M}_1$ and there are no FTL inertial particles in 
$\mathfrak{M}_2$. We are going to prove this simultaneously for 
the two models.}

\hosszu{The field reduct of both models is the Euclidean  field
${\mathfrak Q}$. Thus \ax{AxEField} is valid in the models.}

\hosszu{By (\ref{roe}), world-view transformations are Poincar\'e
transformations, hence they are affine transformations and bijections. 
Thus \ax{AxEv} is valid in models $\mathfrak{M}_1$ and $\mathfrak{M}_2$.}

\hosszu{By (\ref{taxis}), we have that $\wl_k(k)=\taxis$. Thus,
\ax{AxSelf} is valid in the two models, by (\ref{W-d}).}

\hosszu{We say that $\ray{\vx\vy}$ is {\bf light-like} if{}f $\vx\neq\vy$ 
and $\sqspace(\vx,\vy)=\timed(\vx,\vy)$ and it is
{\bf time-like} if{}f  $\sqspace(\vx,\vy)<\timed(\vx,\vy)$.}

\hosszu{By (\ref{wlray}) and by the properties of Poincar\'e transformations,
the world-lines of bodies in $\Ph$ are the light-like rays according to
any observer. Thus the ``speed of light'' is $1$ for any observer.
Thus \ax{AxLight} and the second part of \ax{AxSymD} holds.}

\hosszu{Any Poincar\'e transformation $P$ preserves the spatial distance of
points $\vx,\vy\in\Q^d$ for which $x_1=y_1$ and
$P(\vx)_1=P(\vy)_1$. Therefore, inertial observers agree as to the
spatial distance between two events if these two events are
simultaneous for both of them. We have already shown that the speed of
light is 1 according to each inertial observer in models  $\mathfrak{M}_1$ and
$\mathfrak{M}_2$. Consequently, axiom 
\ax{AxSymD} is also valid in these models.}

\hosszu{To prove that axiom schema \ax{SPR^+} is valid in
models $\mathfrak{M}_1$ and $\mathfrak{M}_2$, it is enough to show
that, for observer
 $k$,  map $\alpha$ given in (\ref{alpha}) is an automorphism. 
It is clear that $\alpha$ leaves the elements of 
\Q\ fixed and it is a permutation on  sets \B, \IOb, $\Ph$  and 
\Ip\  since \IOb\ is the set of Poincar\'e
transformations.  Thus $\B(b)\Leftrightarrow \B\big(\alpha(b)\big)$, 
$\IOb(h)\Leftrightarrow  \IOb\big(\alpha(h)\big)$ 
and $\Ph(p)\Leftrightarrow\Ph(\alpha(p))$.
To prove that $\alpha$ is
an automorphism it remains to prove that  
$\W(k,b,\vx)\Leftrightarrow\W\big(\alpha(k),\alpha(b),\vx)$
and $\M(k,b,q)\Leftrightarrow\M(\alpha(k),\alpha(b),q)$.  
By (\ref{W-d}) and (\ref{M-d}), it is sufficient to prove that
for every inertial observers $h$ and $o$, 
and inertial particle $\vx\vy$,}
\begin{equation}
\label{elso}
\hosszu{\wl_h(o)=\wl_{\alpha(h)}\big(\alpha(o)\big),\ \
\ \wl_h(\vx\vy)=\wl_{\alpha(h)}\big(\alpha(\vx\vy)\big),\
\text{ and }\ } 
\hosszu{\m_h(\vx\vy)=\m_{\alpha(h)}\big(\alpha(\vx\vy)\big).}
\end{equation}
\hosszu{By (\ref{wl-i}), (\ref{wl-d}) 
and (\ref{alpha}), we have}
\begin{multline*}
\hosszu{\wl_{\alpha(h)}\big(\alpha(o)\big)=\alpha(h)
\big[\wl\big(\alpha(o)\big)
\big]= \alpha(h)\big[\alpha(o)^{-1}[\taxis]\big]=}\\
\hosszu{h\circ k^{-1}\big[k\circ o^{-1}[\taxis]\big]=
h\big[ o^{-1}[\taxis]\big]=h[\wl(o)]=\wl_h(o)}, 
\end{multline*}

\begin{multline*}
\hosszu{\wl_{\alpha(h)}\big(\alpha(\vx\vy)\big)=\alpha(h)\big[\wl
\big(\alpha(\vx\vy)
\big)\big]=h\circ k^{-1}\big[\wl\big(k(x)k(y)\big)\big]=}\\
\hosszu{h\circ k^{-1}\big[\ray{k(\vx)k(\vy)}\,\big]=h\circ k^{-1}\circ k
\big[\ray{\vx\vy}\,\big]=h\big[\ray{\vx\vy}\,\big] = \wl_h(\vx\vy),\text{
and}}
\end{multline*}

\begin{multline*}
\hosszu{\m_{\alpha(h)}\big(\alpha(\vx\vy)\big)=\m_{h\circ
k^{-1}}\big(k(\vx)k(\vy)\big)=}\\ 
\hosszu{\timed\big(h\circ k^{-1}\circ k (\vx), h\circ k^{-1}\circ k(\vy)\big)=
\timed\big(h(\vx),h(\vy)\big)=\m_h(\vx\vy).}  
\end{multline*}
\hosszu{Thus (\ref{elso}) above holds.} 

\hosszu{Therefore, axiom schema \ax{SPR^+} is valid in the
models. We have proved that ${\mathfrak M}_1$ and ${\mathfrak M}_2$
are models of the axiom system  \ax{SR}  of kinematics of
special relativity.}

\hosszu{By (\ref{wlray}) and by the properties of Poincar\'e transformations, 
the world-lines of bodies in $\Ip_1$ are the rays, 
the world-lines of bodies in $\Ip_2$ are the time-like and the light-like 
rays according to any observer.
Thus there are massive FTL inertial particles in ${\mathfrak M}_1$
and there are no FTL inertial particles in ${\mathfrak M}_2$.  
Thus we have proved that ${\mathfrak M}_1\models\exists\FTL \Ip$ and
${\mathfrak M}_2\models\neg\exists\FTL \Ip$.}

\hosszu{It remains to prove that axioms of dynamics 
are also valid in the models.}

\hosszu{First we turn proving that, for every natural number $n$, \ax{AxColl_n} 
is valid in the model. The  proof is illustrated in Fig.\ref{dyn7}.
Throughout $\vo\leteq\langle 0,\ldots,0\rangle\in\Q^d$ denotes the origin of the
coordinate system.}
\begin{figure}
\small
\psfrag*{Wi}[l][l]{worldview of observer $k$}
\psfrag*{Wk}[l][l]{worldview of observer $h$}
\psfrag*{x}[lb][lb]{$k(\vx)$}
\psfrag*{y}[lb][lb]{$k(\vy^1)$}
\psfrag*{y2}[rt][rt]{$k(\vy^2)$}
\psfrag*{y3}[rb][rb]{$k(\vy^3)$}
\psfrag*{kx}[lt][lt]{$h(\vx)$}
\psfrag*{ky}[lb][lb]{$h(\vy^1)$}
\psfrag*{ky2}[rb][rb]{$h(\vy^2)$}
\psfrag*{ky3}[rb][rb]{$h(\vy^3)$}

\psfrag*{text1}[t][t]{\shortstack[c]
{$\poscoll_k(\vx\vy^1\,\vx\vy^2\,\vx\vy^3)$\\
$\Updownarrow$\\
$\fvp_k(\vx\vy^2)+\fvp_k(\vx\vy^3)=\fvp_k(\vx\vy^1)$}
}
\psfrag*{text2}[t][t]{\shortstack[c]
{$\poscoll_h(\vx\vy^1\,\vx\vy^2\,\vx\vy^3)$\\
$\Updownarrow$\\
$\fvp_h(\vx\vy^3)=\fvp_h(\vx\vy^1)+\fvp_h(\vx\vy^2)$}
}

\psfrag*{text3}[b][b]{\shortstack[c]{$-P_k(\vx\vy^1)+
                                       P_k(\vx\vy^2)+
                                       P_k(\vx\vy^3)=\vo$\\
                                       $\Updownarrow$\\
                                       $\sum\big(k(\vx)-k(\vy^i)\big)=\vo$}
}

\psfrag*{text4}[b][b]{\shortstack[c]{$-P_h(\vx\vy^1)-
                                       P_h(\vx\vy^2)+
                                       P_h(\vx\vy^3)=\vo$\\
                                       $\Updownarrow$\\
                                       $\sum\big(h(\vx)-h(\vy^i)\big)=\vo$}
}

\psfrag*{P}[lb][lb]{$P_k(\vx\vy^1)$}
\psfrag*{P2}[rt][rr]{$P_k(\vx\vy^2)$}
\psfrag*{P3}[rb][rb]{$P_k(\vx\vy^3)$}

\psfrag*{PP}[lb][lb]{$P_h(\vx\vy^1)$}
\psfrag*{PP2}[lt][lr]{$P_h(\vx\vy^2)$}
\psfrag*{PP3}[rb][rb]{$P_h(\vx\vy^3)$}
\includegraphics[keepaspectratio,width=0.8\textwidth]{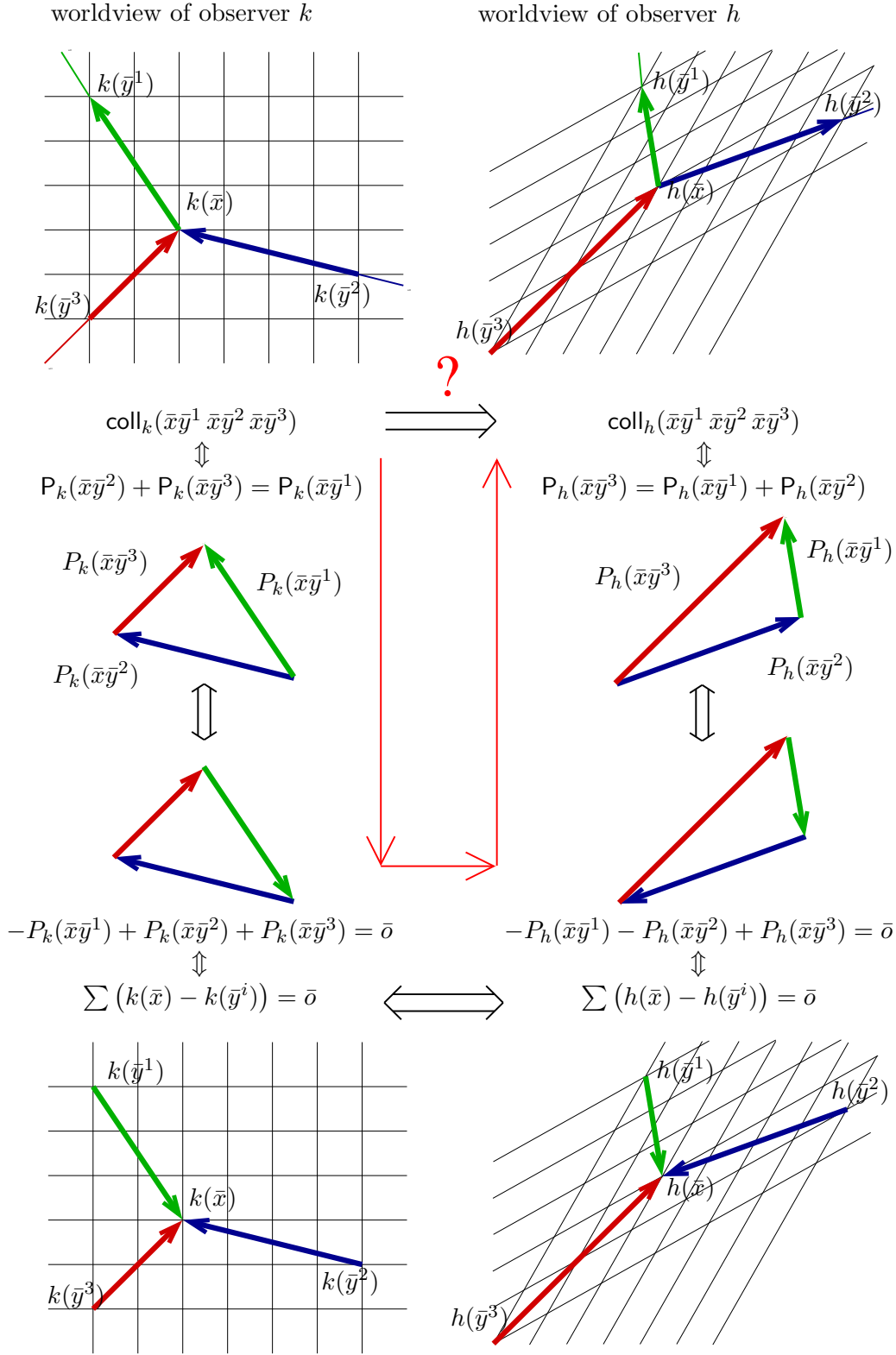}
\caption{\hosszu{Illustration for proving that \ax{AxColl_n} 
is valid in the models}}
\label{dyn7}
\end{figure}
\hosszu{Recall that, by (\ref{wlray}), $\wl_k(\vx\vy)= \ray{k(\vx)k(\vy)}$ 
for every inertial particle $\vx\vy$ and inertial observer $k$. Therefore}
\begin{equation}
\hosszu{\inc_k(\vx\vy,k(\vx))\ \Longleftrightarrow\
 k(y)_1< k(x)_1\quad \text{ and }}\quad
\hosszu{\out_k(\vx\vy,k(\vx))\  \Longleftrightarrow\  k(x)_1< k(y)_1.}
\label{inout-e}
\end{equation}

\hosszu{Recall that, by (\ref{harmadikcs}), $\fvp_k(\vx\vy)=k(\vx)-k(\vy)$ 
if{}f
$k(\vy)_1<k(\vx)_1$ and $\fvp_k(\vx\vy)=k(\vy)-k(\vx)$ if{}f
$k(\vx)_1<k(\vy)_1$.
Now, by  (\ref{harmadikcs}) and (\ref{inout-e}),
for every inertial observer $k$ and inertial particles $\vx^1\vy^1,\ldots
\vx^n\vy^n$ with $\bigwedge_{i=1}^n \sqspeed_k(\vx^i\vy^i)<\infty$,
we get}
\begin{multline}
\label{szornyu1} 
\hosszu{\sum_{\{i\: :\: \inc_k (\vx^i\vy^i,k(\vx^i))\}}\fvp_k(\vx^i\vy^i)=
\sum_{\{i\: :\: \out_k (\vx^i\vy^i,k(\vx^i))\}}\fvp_k(\vx^i\vy^i)\qquad
\stackrel{(\ref{inout-e})}{\Longleftrightarrow}}\\
\hosszu{\sum_{\{i\: :\: k(\vy^i)_1< k(\vx^i)_1\}}\fvp_k(\vx^i\vy^i)=
\sum_{\{i\: :\: k(\vx^i)_1< k(\vy^i)_1\}}\fvp_k(\vx^i\vy^i)\qquad
\stackrel{(\ref{harmadikcs})}{\Longleftrightarrow}}\\
\hosszu{\sum_{\{i\: :\: k(\vy^i)_1< k(\vx^i)_1\}} \big(k(\vx^i)-k(\vy^i)\big)=
\sum_{\{i\: :\: k(\vx^i)_1< k(\vy^i)_1\}} \big(k(\vy^i)-k(\vx^i)\big)
\qquad {\Longleftrightarrow}}\\ 
\hosszu{\sum_{i=1}^n \big(k(\vx^i)-k(\vy^i)\big)=\vo.}
\end{multline}

\hosszu{Therefore,  by (\ref{szornyu1}) 
and the equivalent form (\ref{pcollfm-d})  
of the definition of $\poscoll$, we get that}
\begin{multline}
\label{coll1}
\hosszu{\poscoll_k(\vx^1\vy^1\ldots\vx^n\vy^n)\ \Longleftrightarrow\
\vx^1=\ldots=\vx^n\land}\\
\hosszu{\bigwedge_{i=1}^n \sqspeed_k(\vx^i\vy^i)<\infty\  \land\  
\sum_{i=1}^n\big(k(\vx^i)-k(\vy^i)\big)=\vo.} 
\end{multline}
\hosszu{for every $n$, inertial observer $k$, and
inertial particles $\vx^1\vy^1,\ldots,\vx^n\vy^n$.}

\hosszu{To prove that \ax{AxColl_n} is valid in the models,
let $k$ and $h$ be inertial observers and let $\vx^1\vy^1,\ldots,\vx^n\vy^n$
be inertial particles such that $\poscoll_k(\vx^1\vy^1\ldots\vx^n\vy^n)$
and $\sqspeed_h(\vx^i\vy^i)<\infty$ for every $i$. Then, by (\ref{coll1}),
$\vx^1=\ldots=\vx^n$  and $\sum_{i=1}^n\big(k(\vx^i)-k(\vy^i)\big)=\vo$.
Therefore, $\sum_{i=1}^n\big(h(\vx^i)-h(\vy^i)\big)=\vo$, since 
$h\circ k^{-1}$ is an affine transformation taking $k(\vx)$ to 
$h(\vx)$ for every $\vx$. 
Thus, by
(\ref{coll1}), $\poscoll_h(\vx^1\vy^1\ldots\vx^n\vy^n)$ holds.
Therefore, \ax{AxColl_n} is valid in the models for every $n$. Thus the axiom
schema \ax{Coll} is valid in the models.}

\hosszu{Next we turn proving that \ax{Ax\forall\inecoll} is valid in the models. 
Let $\fm_1\subseteq\Q^d$ 
be the set of vectors with positive time components and
let $\fm_2\subseteq\Q^d$ be the set of non-FTL vectors with
positive time components, i.e.,}
\begin{eqnarray*}
\hosszu{\fm_1} & \hosszu{\leteq} &  
\hosszu{\{ \vx\in\Q^d\, :\, 0<x_1\},  \text{ and}}\\
\hosszu{\fm_2} &\hosszu{\leteq}&  \hosszu{\{ 
\vx\in\fm_1\, :\, \sqspace(\vx,\vo)\leq \timed(\vx,\vo)\}}.
\end{eqnarray*} 

\hosszu{Let $\fm$ denote $\fm_1$ in the case of $\mathfrak{M}_1$
and denote $\fm_2$ in the case of $\mathfrak{M}_2$.}

\hosszu{Recall that, by (\ref{wlray}), $\wl_k(\vx\vy)=\ray{k(\vx)k(\vy)}$.  
By (\ref{harmadikcs}), $\fvp_k(\vx\vy)=k(\vx)-k(\vy)$
if{}f $\vx\vy$ is incoming at $k(\vx)$ according
to $k$, and $\fvp_k(\vx\vy)=k(\vy)-k(\vx)$ if{}f
$\vx\vy$ is outgoing at $k(\vx)$. By the above, by the fact that the
observers are Poincar\'e transformations, and by
(\ref{eq-IOb})--(\ref{imp2-d}), it is easy to see that for every
inertial observer at every coordinate point, the set of four-momenta of the
incoming bodies is $\fm$ and the same holds for the outgoing bodies.
Formally, for every inertial observer $k$ and coordinate point $\vx$,}
\begin{equation}
\label{h1}
\hosszu{\fm =\{ \fvp_k(b)\, :\, \Ip(b) \land \inc_k(b,\vx)\}
=\{ \fvp_k(b)\, :\, \Ip(b)\land\out_k(b,\vx)\}. 
}
\end{equation}
\hosszu{By (\ref{wlray}),  for every inertial observer, every inertial particle
of finite speed is an incoming or outgoing body at some coordinate point.
Therefore, (\ref{h1}) implies that}
\begin{equation}
\label{h2}
\hosszu{\fm=\{\fvp_k(b)\, :\, \Ip(b)\land\vv_k(b)<\infty\}}
\end{equation}
\hosszu{for every inertial observer $k$.
 It can be easily seen that $\fmom$ is closed under addition,
i.e.,}
\begin{equation}
\label{h3}
\hosszu{\fvp,\fvp'\in\fmom\quad   \Longrightarrow\quad   
\fvp+\fvp'\in\fmom.}
\end{equation}

\hosszu{To prove that \ax{Ax\forall\inecoll} is valid, let $\vx$ be a 
coordinate point, let $k$ be an inertial observer
and let $a$ and $b$ be inertial particles with finite speeds according to
$k$ such that the sum of their relativistic masses is nonzero.
We have to prove that there are inertial particles $a'$ and $b'$
such that $\vx\text{-}\inecoll_k(a'b')$ and
$\fvp_k(a')=\fvp_k(a)$ and $\fvp_k(b')=\fvp_k(b)$.
 By (\ref{h2}) and (\ref{h3}), we have that 
$\fvp_k(a),\fvp_k(b), \fvp_k(a)+\fvp_k(b)\in\fmom$.
By (\ref{h1}),  there are 
inertial particles  $a',b'$ and $c$ such that
$\fvp_k(a')=\fvp_k(a)$, $\fvp_k(b')=\fvp_k(b)$,
$\fvp_k(c)=\fvp_k(a)+\fvp_k(b)$, $\inc_k(a',\vx)$, $\inc_k(b',\vx)$
and $\out_k(c,\vx)$. Thus $\poscoll_k(a'b'c)$ by (\ref{pcollfm-d}).
Now, by (\ref{xinecoll-d}), $\vx\text{-}\inecoll_k(a'b')$. 
Therefore, \ax{Ax\forall\inecoll} is valid in the models.}

\hosszu{To prove that \ax{AxSpd} is valid, let $k,h\in\IOb$, $\vx\vy\in\Ip$
and $q\in\Q$ be
such that  $\sqspeed_k(\vx\vy)=\sqspeed_h(\vx\vy)=q<1$.
It is enough to prove that $\m_k(\vx\vy)=\m_h(\vx\vy)$.
By (\ref{eeeee}), }
\begin{eqnarray}
\label{klamb}
\hosszu{\sqspace\big(k(\vx),k(\vy)\big)}& \hosszu{=} &  
\hosszu{q\cdot\timed\big(k(\vx),k(\vy)\big),
\text{ and}}\\
\label{hlamb}
\hosszu{\sqspace\big(h(\vx),h(\vy)\big)}&  \hosszu{=} &  
\hosszu{q\cdot\timed\big(h(\vx),h(\vy)\big).}
\end{eqnarray}

\hosszu{Since $h\circ k^{-1}$ is a Poincar\'e transformation taking
$k(\vx)$ and $k(\vy)$ to $h(\vx)$ and $h(\vy)$, respectively,
we get that}
\begin{equation}
\hosszu{\timed\big(k(\vx),k(\vy)\big)^2   -\sqspace\big(k(\vx),k(\vy)\big)^2 
=}
\hosszu{\timed\big(h(\vx),h(\vy)\big)^2-\sqspace\big(h(\vx),h(\vy)\big)^2.}
\label{minkdist}
\end{equation}
\hosszu{By (\ref{klamb})--(\ref{minkdist}),
$(1-q^2)\timed\big(k(\vx),k(\vy)\big)^2=
 (1-q^2)\timed\big(h(\vx),h(\vy)\big)^2$.
Thus $\timed\big(k(\vx),k(\vy)\big)=\timed\big(h(\vx),h(\vy)\big)$.
Now, by (\ref{wl-d}),  
$\m_k(\vx\vy)=\timed\big(k(\vx),k(\vy)\big)=\timed\big(h(\vx),h(\vy)\big)=
\m_h(\vx\vy)$. Thus $\m_k(\vx\vy)=\m_h(\vx\vy)$. 
Therefore,  \ax{AxSpd} is valid in the models.}

\hosszu{To prove that \ax{AxMass} is valid, let $k$ and $h$ be inertial
observers and let $b$ and $b'$ be inertial particles such that
their velocities and their relativistic masses coincide
according to observer $k$. Then, by definition (\ref{fourm-d}) of
four-momentum, the four-momenta of $b$ and $b'$ coincide according to 
observer $k$. Assume that the speed of $b$ is finite according to $h$.
By (\ref{harmadikcs}) and by the fact that $h\circ k^{-1}$ is an
affine transformation, it is easy to prove that the four-momenta 
of $b$ and $b'$ coincide according to inertial observer $h$, too.}%
\footnote{\hosszu{This is so because of the following. Let $b=\vx\vy$ and
$b'=\vx'\vy'$. Then, by (\ref{harmadikcs}), $\pm(k(\vx)-k(\vy))
=\fvp_k(\vx\vy)=\fvp_k(\vx'\vy')=\pm(k(\vx')-k(\vy'))$.
But then $h(\vx)-h(\vy)=\pm (h(\vx')-h(\vy'))$ since
$h\circ k^{-1}$ is an affine transformation taking 
$k(\vx),k(\vy),k(\vx'),k(\vy')$
to  $h(\vx),h(\vy),h(\vx'),h(\vy')$, respectively. By
(\ref{harmadikcs}), we conclude that $\fvp_h(\vx\vy)=\pm\fvp_h(\vx'\vy')$.
The time-components of the four-momenta are positive since relativistic
masses are positive. Therefore, $\fvp_h(\vx\vy)=\fvp_h(\vx'\vy')$.}}    
\hosszu{Then the relativistic masses of $b$ and $b'$ coincide according to $h$
since relativistic mass is the time component of the four-momentum.
Therefore, \ax{AxMass} is valid in the models.}

\hosszu{Now we turn proving that \ax{AxThEx^+} is valid in the models.
We say that (straight) line $\{\vx+q\cdot(\vy-\vx)\, :\, \Q(q)\}$
 is {\bf time-like} if{}f $\ray{\vx\vy}$ is time-like.
The world-lines of inertial observers are the time-like lines
according to observer $\Id$
by (\ref{wl-i}) and (\ref{wl-d}) since  Poincar\'e transformations
take the $\taxis$ to time-like lines and  for any time-like line $\ell$
there is an orthocronous Poincar\'e transformation taking $\ell$ to
$\taxis$. Poincar\'e transformations take the set of time-like lines onto
the set of time-like lines. Therefore, the world-lines of inertial observers
are the time-like lines according to any observer by (\ref{wl-d}).
Therefore, the first part of \ax{AxThEx^+} holds.
It is easy to see that the second part of \ax{AxThEx^+} holds, because of 
(i)--(iv) below.  (i) The set of four-momenta of the incoming bodies 
contains set $\fm_2$ by (\ref{h1}).
(ii) The time components of the four-momenta are the
relativistic masses. (iii) Four-momenta are parallel to the world-lines.
(iv) World-lines of inertial particles are rays.
Therefore, \ax{AxThEx^+} is valid in the models.}

To prove that \ax{Ax\forall Coll} is valid in the models,
let $k$ be an inertial observer, $\vx$ be a coordinate point,
and $a\leteq \vx'\vy'$ be an inertial particle of finite speed according to 
$k$. Let 
$b_1\leteq k^{-1}(\vx)\big(k^{-1}(\vx)+\vx'-\vy'\big)$ and 
$b_2\leteq k^{-1}(\vx)\big(k^{-1}(\vx)+\vy'-\vx'\big)$. By the fact that
$k$ is an affine transformation, 
$k\big(k^{-1}(\vx)+\vx'-\vy'\big)=
\vx+k(\vx')-k(\vy')$ and  $k\big(k^{-1}(\vx)+\vy'-\vx'\big)=
\vx+k(\vy')-k(\vx')$. 
Thus, by  \eqref{wlray},
$\wl_k(a)=\ray k(\vx')k(\vy')$,
$\wl_k(b_1)=\ray{\vx\big(\vx+k(\vx')-k(\vy')\big)}$ and 
$\wl_k(b_2)=\ray{\vx\big(\vx+k(\vy')-k(\vx')\big)}$. Hence the velocities of
$a$, $b_1$ and $b_2$ coincide, and one of $b_1$ and $b_2$ is incoming
at $\vx$ and the other one is outgoing at $\vx$ according to observer $k$. 
Furthermore, by \eqref{wl-d}, 
$\m_k(a)=\m_k(b_1)=\m_k(b_2)=\timed\big(k(\vx'),k(\vy')\big)$.
Thus \ax{Ax\forall Coll} is valid in the models.   
 
\hosszu{By the above, axioms of dynamics are also valid in
our models. Therefore, both $\mathfrak{M}_1$ and ${\mathfrak M}_2$ are
models of \ax{SRDyn}. This completes the proof.}
\end{proof}


\section{Concluding remarks} \label{sec-con} Paper \cite{FTLconsSR} shows
that the existence of FTL particles is logically independent of \ax{SR} an
axiom system of special relativistic kinematics based on Einstein's original
postulates.  In this paper, we have seen that the existence of massive FTL
inertial particles is logically independent of \ax{SRDyn} an extension of
\ax{SR} to special relativistic dynamics.

\begin{figure}[h!tb] \begin{center} \tikzset{->-/.style={decoration={
  markings, mark=at position #1 with {\arrow{>}}},postaction={decorate}}}
\begin{tikzpicture}[scale=2] \coordinate (o) at (0,0); \draw[gray] (-2,-1)
grid (2,2); \draw[red, very thick] (1,-1) -- (o) -- (2,2); \draw[red, very
thick] (-1,-1) -- (o) -- (-2,2); \draw[magenta, ultra thick] (-1.95,2) .. 
controls (-1.8,1.93) and (-1.2,1) ..  (0,1) ..  controls (1.2,1) and
(1.8,1.93) ..  (1.95,2); \draw[magenta, ultra thick] (-1,0.04) ..  controls
(-1.1,1.2) and (-1.9,1.8) ..  (-2,1.95); \draw[magenta, ultra thick]
(1,0.04) ..  controls (1.1,1.2) and (1.9,1.8) ..  (2,1.95); \draw[magenta]
(1,0) circle (0.04); \draw[magenta] (-1,0) circle (0.04);
\draw[->-=.75,blue, ultra thick,>=stealth] (o) -- (1.2,0.8);
\draw[dashed, thin] (1.2,0) -- (1.2,0.8) -- (0,0.8) node[left]
{$\m_{k}(b)$}; \draw[dashed, thin] (1.2,0) -- (1.2,0.8) -- (0,0.8); \fill
(1.2,0.8) circle (0.04); \fill (1.2,0) circle (0.04); \fill (0,0.8) circle
(0.04); \fill (o) circle (0.04); \end{tikzpicture} \end{center}
\caption{\label{fig-decrease} Illustration for equation (\ref{inv})}
\end{figure}
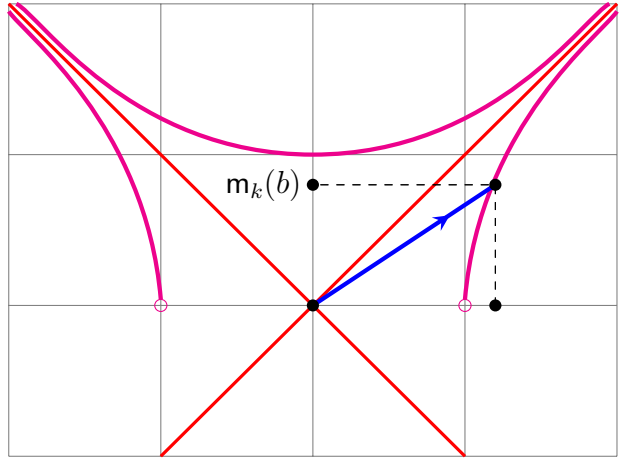

In \cite{SIGMA}, we show that \ax{SRDyn} gives new predictions
on relativistic masses of FTL inertial particles.  In more detail,
\ax{SRDyn} implies that \begin{equation}\label{inv}
\m_k(b)\sqrt{\left|1-\vv_k(b)^2\right|}=\m_h(b)\sqrt{\left|1-\vv_h(b)^2\right|},
\end{equation} where $b$ is a possibly FTL inertial particle and $k$ and $h$
are (ordinary slower than light) inertial observers.  Equation (\ref{inv})
gives back the usual mass-increase theorem for slower than light particles,
and predicts that the relativistic mass and momentum of an FTL particle
decrease with the speed, see Fig.\ref{fig-decrease}.

Similar predictions on FTL particles appear in
Bilaniuk-Deshpande-Su\-dar\-shan~\cite{BiDeSu62}, Sudarshan~\cite{Su70},
Recami~\cite{Rec86,recami-ftl}, and Hill-Cox~\cite{HC12}.

The results in \cite{SIGMA} show that the construction we used here is the
only possible way for extending a model of relativistic dynamics with FTL
particles if some natural basic assumptions (such as conservation of
relativistic mass and momenta) are assumed.

\section{Acknowledgment} We are grateful for valuable discussions to Hajnal
Andr\'eka, G\'abor Czimer, M\'arton G\"om\"ori, Attila Moln\'ar, Istv\'an
N\'emeti, Attila Seress, Ren\'ata Tordai.  This research is supported by the
Hungarian Scientific Research Fund for basic research grants No.~T81188 and
No.~PD84093 as well as by Bolyai Grant.

\bibliographystyle{plain} 

\end{document}